\documentclass[aps,prx,reprint,amsmath,amssymb,superscriptaddress,longbibliography]{revtex4-2}

\usepackage[utf8]{inputenc}
\usepackage[shortlabels]{enumitem}
\usepackage{textcomp}
\usepackage{graphicx}
\usepackage[normalem]{ulem}
\usepackage{booktabs,colortbl}
\usepackage[table]{xcolor}
\definecolor{myblue}{RGB}{95, 141, 211}

\usepackage{hyperref}
\hypersetup{
  colorlinks,%
  citecolor=black,%
  filecolor=black,%
  linkcolor=black,%
  urlcolor=black
}
\usepackage[all]{hypcap}

\usepackage{amsthm}
\newtheorem{theorem}{Theorem}
\newtheorem{proposition}[theorem]{Proposition}
\newtheorem{lemma}[theorem]{Lemma}
\newtheorem{corollary}[theorem]{Corollary}
\newtheorem{result}{Result}

\renewcommand*{\vec}[1]{\mathbf{#1}}
\newcommand*{\bra}[1]{\ensuremath{\langle #1 \vert}}
\newcommand*{\ket}[1]{\ensuremath{\vert #1 \rangle}}

\newcommand*{\di}{\mathrm{d}}
\newcommand*{\comm}[2]{\left[ #1,#2 \right]}

\newcommand*{\tr}{\mathrm{tr}}
\newcommand*{\supp}{\mathrm{supp}}
\newcommand*{\im}{\mathrm{Im}}

\newcommand*{\norm}[1]{\left\| #1 \right\|}
\newcommand*{\hil}{\mathcal{H}}
\renewcommand*{\eqref}[1]{Eq.~(\ref{#1})}

\newcommand*{\figref}[1]{Fig.~\ref{#1}}
\newcommand*{\Figref}[1]{Figure~\ref{#1}}

\newcommand{\clus}[1]{{\mathbf{#1}}}

\newcommand{\lc}{ \nu  }

\begin{document}

\title{Classical simulation of short-time quantum dynamics}
\date{\today}

\author{Dominik S. Wild}
\email{dominik.wild@mpq.mpg.de}
\affiliation{Max-Planck-Institut für Quantenoptik, Hans-Kopfermann-Straße 1, D-85748 Garching, Germany}

\author{\'Alvaro M. Alhambra}
\email{alvaro.alhambra@csic.es}
\affiliation{Max-Planck-Institut für Quantenoptik, Hans-Kopfermann-Straße 1, D-85748 Garching, Germany}
\affiliation{
Instituto de F\'isica Te\'orica UAM/CSIC, C/ Nicol\'as Cabrera 13-15, Cantoblanco, 28049 Madrid, Spain}

\begin{abstract} 
Recent progress in the development of quantum technologies has enabled the direct investigation of dynamics of increasingly complex quantum many-body systems. This motivates the study of the complexity of classical algorithms for this problem in order to benchmark quantum simulators and to delineate the regime of quantum advantage. Here we present classical algorithms for approximating the dynamics of local observables and nonlocal quantities such as the Loschmidt echo, where the evolution is governed by a local Hamiltonian. For short times, their computational cost scales polynomially with the system size and the inverse of the approximation error. In the case of local observables, the proposed algorithm has a better dependence on the approximation error than algorithms based on the Lieb--Robinson bound. Our results use cluster expansion techniques adapted to the dynamical setting, for which we give a novel proof of their convergence. This has important physical consequences besides our efficient algorithms. In particular, we establish a novel quantum speed limit, a bound on dynamical phase transitions, and a concentration bound for product states evolved for short times.


\end{abstract}

\maketitle

\section{Introduction}

The study ofthe dynamics of quantum many-body models is a highly active area of research in quantum information science, both from the perspective of physics and of computation. Probing dynamics provides access to a wealth of physical phenomena and can enable the solution of hard computational problems. Existing quantum simulators are already allowing us to explore quantum dynamics of high complexity. To assess the potential for quantum speed-up, it is important to understand the reach of classical methods for these tasks. Unless all quantum computations can be classically simulated, i.e.~\textsf{BPP} = \textsf{BQP}, classical algorithms will be unable to approximate quantum dynamics in an arbitrary setting. Nevertheless, there exist restricted regimes in which efficient classical simulation is possible. An important special case is evolution for short times, during which the quantum information will not spread much. Tensor-network methods illustrate this point, as they provide provably efficient means of simulating unitary dynamics in one spatial dimension for short times \cite{Osborne2006,KuwaharaImproved}.

In this work, we characterize the computational complexity of short-time dynamics under a local Hamiltonian more generally. Locality in this context means that the Hamiltonian can be written as a sum of operators supported on small subsystems. We do not require geometric locality and we do not restrict the Hamiltonian to a finite-dimensional lattice. Instead, we only impose that every term in the Hamiltonian overlaps with a constant number of other terms. These conditions are satisfied by a wide class of physically relevant Hamiltonians, including parent Hamiltonians of quantum LDPC error-correcting codes~\cite{breuckmann2021}. Dynamics governed by Hamiltonians of this type are amenable to \emph{cluster expansions}, which have long been used in both classical and quantum statistical mechanics of lattice models~\cite{RuelleBook,friedli2017,Malyshev_1980,Park1982,KoteckyPreiss}, leading to results such as the uniqueness of Gibbs states~\cite{Dobrushin1987,dobrushin1996estimates}, efficient approximation schemes for partition functions~\cite{mann2021,harrow2020}, the decay of correlations~\cite{Ueltschi2004,kliesch2014,Frohlich_2015}, and concentration bounds~\cite{Neto_n__2004,kuwahara2020_gaussian}. Despite their long history, cluster expansions have typically been applied to equilibrium properties, while dynamics have been rarely considered.

The paper is structured as follows. In the remainder of the introduction, we summarize the main results and give an overview of their implications for the complexity of quantum dynamics. In Sec.~\ref{sec:setup}, we define the cluster notation used throughout. Sec.~\ref{sec:obs} discusses the results and algorithm for local observables. In Sec. \ref{sec:loschmidt}, we describe corresponding results for the Loschmidt echo. Physical consequences of these results are discussed in Sec. \ref{sec:implications}. We conclude in Sec.~\ref{sec:conclusion} with further remarks and open questions. The main text provides an overview of all the proof techniques, while technical details that are less crucial to the understanding are placed in appendices.

\setlength{\tabcolsep}{8pt}

\begin{table*}[t]\label{ta:comp}
    \begin{tabular}{lllll}
        \toprule
        & $t < t^*, t^*_L$ & $t = \mathcal{O}(1)$ & $t = \mathcal{O}(\text{polylog}(n))$ & $t = \mathcal{O}(\text{poly}(n))$ \\
        \midrule
        $\langle A(t) \rangle$  & {\color{myblue} \textsf{P}} &  {\color{myblue} \textsf{P}} & ? & \textsf{BQP}-complete \cite{janzing2005} \\
        $\log \langle e^{-iHt} \rangle$  & {\color{myblue} \textsf{P}} & \textsf{\#P}-hard \cite{galanis2021,galanis2022} & \textsf{\#P}-hard &  \textsf{\#P}-hard \\
        $ \langle e^{-iHt} \rangle$  & {\color{myblue} \textsf{P}} & ? & ? &  \textsf{BQP}-complete \cite{de_las_cuevas2011} \\
        \bottomrule
    \end{tabular}
    \caption{The computational complexity of computing the quantities in the leftmost column with additive error $\varepsilon = 1/\text{poly}(n)$, where $n$ is the system size, for different times. All expectation values are with respect to product initial states. The constants $t^*$ and $t_L^*$ are independent of the system size but depend on the details of the problem. Entries highlighted in blue are results from this work. Question marks denote regimes where the computational complexity is unknown.}
\end{table*}

\subsection{Summary of results}

Our first main result concerns the dynamics of a few-body observable $A$ under a local Hamiltonian $H$.
\begin{result}\label{re:result1}
    (Informal version of Theorem~\ref{th:algoObs}) Given a local Hamiltonian $H$, a few-body operator $A$, and a product state $\rho$, there exists an algorithm that approximates $\tr( e^{iHt}A e^{-iHt}\rho)$ up to additive error $\varepsilon$ with run time of at most
    \begin{equation}
            \mathrm{poly} \left[ \left( \frac{1}{\varepsilon} e^{\pi |t| / t^*} \right)^{\exp({\pi |t| /  t^*})} \right],
        \end{equation}
    where $t^*$ is a positive constant.
\end{result}

The scaling with the time $t$ is rather unfavorable, but the computational cost is independent of system size. Moreover, for any constant value of $t$, the run time has a polynomial dependence on $1/\varepsilon$, which is an improvement over all previously known algorithms, such as those based on Lieb--Robinson bounds. The cluster expansion can thus be seen as an alternative approach to analyzing the effective locality and light-cone structure of many-body dynamics.

Our second result characterizes the complexity of computing the Loschmidt echo, $\tr(e^{-itH} \rho)$.
\begin{result}\label{re:resultLE}
    (Informal version of Theorem~\ref{th:computation}) Given a local Hamiltonian $H$, a product state $\rho$, and $\vert t\vert<t^*_L$, there exists an algorithm that approximates $\log \tr(e^{-iHt} \rho)$ up to additive error $\varepsilon$ with run time of at most
    \begin{equation}
        n \times \mathrm{poly}\left[ \left( \frac{n}{1 - |t|/t_L^*} \frac{1}{\varepsilon} \right)^{1/\log(t_L^*/|t|)} \right],
    \end{equation}
    where $t^*_L$ is a positive constant.
\end{result}

A key insight of this work is that when $\rho$ is a product state, the objects analyzed in Results~\ref{re:result1} and~\ref{re:resultLE} fit the framework of cluster expansions that are commonly applied to the partition function $\tr(e^{-\beta H})$~\cite{kuwahara2020_gaussian,mann2021,haah2021}. We give a novel proof of the convergence of these expansions based on the counting of trees. 

In both algorithms, the cluster expansion enables an efficient grouping of the terms of a Taylor series by clusters of subsystems. Crucially, we show that only connected clusters contribute. Because the number of connected clusters grows at most exponentially with the size of the cluster, we can establish the convergence of the cluster expansion at short times by bounding the magnitude of the individual terms. The computational cost of the approximation algorithms follows by estimating the cost of computing a truncated cluster expansion while controlling the truncation error. For local observables, we are able to extend the algorithm beyond the radius of convergence of the cluster expansion using analytic continuation. The doubly exponential dependence on the evolution time $t$ in Result~\ref{re:result1} is a direct consequence of the analytic continuation scheme.

The above results have several important physics implications. Result~\ref{re:resultLE} implies that dynamical phase transitions~\cite{heyl2018} cannot occur at times $t \leq t_L^*$ for local Hamiltonians and product initial states. In addition, it establishes a novel quantum speed limit~\cite{Deffner_2017} that is independent of system size, in stark contrast to previous results for general initial states~\cite{mandelstam1991,margolus1998}. A generalization of Result~\ref{re:resultLE} to the multi-Hamiltonian Loschmidt echo $\tr(\rho \prod_l e^{-itH^{(l)}})$, with $H^{(l)}$ local Hamiltonians, allows us to prove Gaussian concentration bounds of local observables on states evolved for a short time, a case not covered by previous results~\cite{Kuwahara_2016,Anshu_2016,AnshuConc2022}. More precisely, we show that the probability of measuring $H^{(2)}$ in the evolved product state $\rho(t)=e^{-itH^{(1)}} \rho e^{itH^{(1)}}$ away from the mean $\tr(\rho(t) H^{(2)} )$ by $\delta$ is suppressed by $e^{\mathcal{O}(-\delta^2 / n)}$.

\subsection{Complexity of dynamics} 

Our main results have nontrivial consequences for the computational complexity of short-time quantum dynamics, which are summarized in Table~\ref{ta:comp}. For observables, it is known that approximating $\langle A (t) \rangle$ with additive error $\varepsilon = 1/\text{poly}(n)$ up to times $t = \text{poly}(n)$ is \textsf{BQP}-complete. This follows from the fact that determining the state of a single qubit at the output of a circuit with $\text{poly}(n)$ gates is \textsf{BQP}-complete, combined with the existence of local Hamiltonians that simulate arbitrary quantum computations~\cite{janzing2005}. At the same time, Theorem~\ref{th:algoObs} shows that we can compute $\langle A (t) \rangle$ classically with a similar error with computational cost $\text{poly}(n)$ as long as $t =\mathcal{O}(1)$. This indicates a transition in the complexity of simulating local observables as the system evolves. The exact nature of this transition and the computational complexity of the intermediate regime $t \sim \text{polylog}(n)$, where the cluster expansion fails to be efficient, remains an open problem.

For the Loschmidt echo, there are two meaningful notions of approximation. The first one is with a small multiplicative error, which is equivalent to a small additive error in the logarithm. For this, Theorem \ref{th:computation} shows that there is an efficient approximation algorithm for times $ \vert t\vert <t^*_L$ with polynomial cost in both $n$ and $1/\varepsilon$. Unlike in Theorem~\ref{th:algoObs}, it is not possible to extend this result to arbitrary times. If we take $\rho \propto \mathbb{I}$, $L(t)$ becomes an imaginary-time partition function. Approximating this for $\vert t \vert =\mathcal{O}(1)$ even with an $\mathcal{O}(1)$ multiplicative error has been shown to be \textsf{\#P}-hard for 2-local, classical Ising models~\cite{galanis2022}. Hence, there is only a constant gap between the times accessible with our algorithm and the \textsf{\#P}-hard regime, where an efficient algorithm is unlikely to exist. The complexity of the analogous problem for real partition functions and the thermodynamic free energy has been recently considered in Ref.~\cite{Bravyi2022}.

We may also consider the weaker additive approximation to the Loschmidt echo. Since $\vert L(t) \vert \le 1$, Theorem \ref{th:computation} also implies that we can efficiently approximate the Loschmidt echo to an $\varepsilon$-additive error for $\vert t \vert < t^*_L$. Calculating the Loschmidt echo for circuits of polynomial size with additive error $\varepsilon = 1/\text{poly}(n)$ is \textsf{BQP}-complete~\cite{de_las_cuevas2011}. To the best of our knowledge, the intermediate regime has not been explored. 

\section{Setup}\label{sec:setup}
\subsection{Hamiltonian}

We consider a set of $n$ spins, $V$. Each spin $v \in V$ is associated with a local Hilbert space $\hil^v$ with $\dim \hil^v = d$. The total Hilbert space is formed by the tensor product space $\hil = \bigotimes_{v \in V} \hil^v$. We call a subset of spins $X \subseteq V$ a subsystem. For any linear operator $A$ on $\hil$, we denote its support by $X = \supp (A)$, i.e., $X$ is the smallest subsystem in $V$ on which $A$ acts nontrivially.

Next, we formally define the notion of local Hamiltonians. Given a set of subsystems $S$, we write a Hamiltonian $H$ as
\begin{equation}
    H = \sum_{X \in S} \lambda_X h_X,
\end{equation}
where each $\lambda_X$ is a real coefficient and $h_X$ is a Hermitian operator acting on the subsystem $X$ such that $\supp ( h_X ) = X$. The coefficients satisfy $|\lambda_X| \leq 1$ and are chosen such that $\Vert h_X \Vert = 1$, where $\Vert \cdot \Vert$ is the operator norm. A Hamiltonian is called $k$-local if it is a sum of terms that act on at most $k$ sites or, equivalently, $|X| \leq k$ for all $X \in S$. 

To characterize the connectivity of the Hamiltonian, we define the associated interaction graph $G$~\cite{haah2021}. Given a set of subsystems $S$, the interaction graph $G$ is a simple graph with vertex set $S$. There is an edge between two vertices $X$ and $Y$ if the respective subsystems overlap. We denote the maximum degree of the interaction graph by $\mathfrak{d}$. Throughout this work, we only consider $k$-local Hamiltonians for which, in addition, $\mathfrak{d}$ is independent of the system size $n$. Each local term in the Hamiltonian therefore only overlaps with a constant number of other terms, which includes many physically relevant cases such as Hamiltonians with finite-range interactions. We point out that the number of terms $|S|$ in these Hamiltonians increases at most linearly with the number of spins $n$.

\subsection{Clusters} \label{sec:clusters}
We define a cluster as a nonempty multiset of subsystems from $S$. Here, multiset refers to a set with possibly repeated elements but without ordering. We use bold-font letters $\vec{V}, \vec{W}, \ldots$ to denote clusters. We call the number of times a subsystem $X$ appears in a cluster $\vec{W}$ the multiplicity $\mu_\vec{W}(X)$. If $X$ is not contained in $\vec{W}$, then $\mu_\vec{W}(X) = 0$. The size $|\vec{W}| = \sum_{X \in S} \mu_\vec{W}(X)$ of a cluster is the number of subsystems that it contains, including their multiplicity. The set of all clusters of size $m$ is denoted by $\mathcal{C}_m$ and the set of all clusters by $\mathcal{C} = \bigcup_{m \geq 1} \mathcal{C}_m$.

\begin{figure}[t]
    \centering
    \includegraphics[width=\columnwidth]{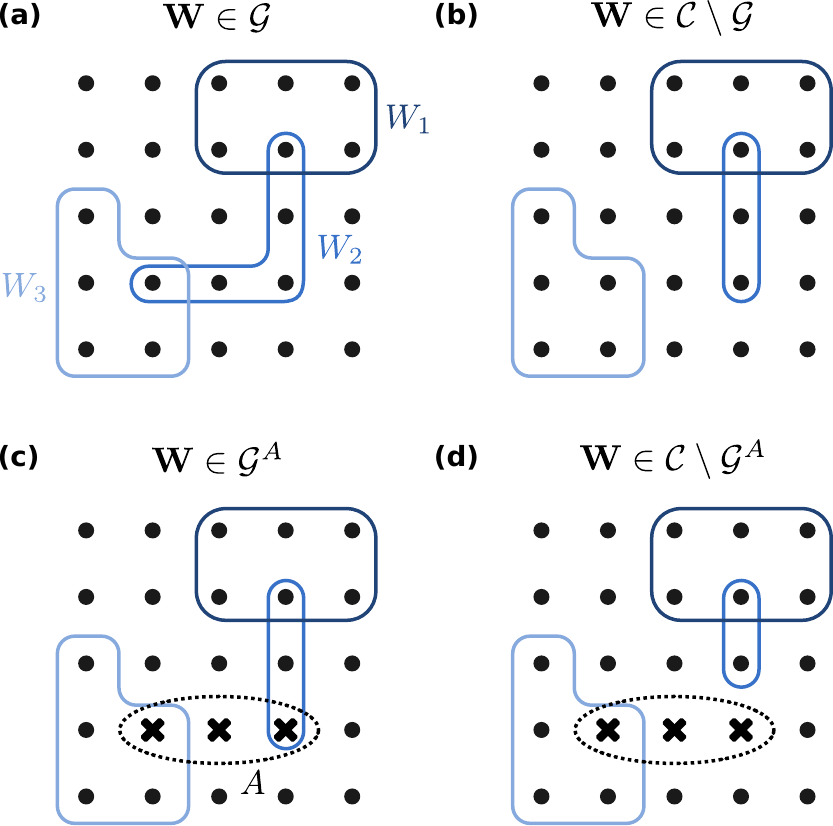}
    \caption{A cluster of three subsystems $\clus{W} = \{ W_1,W_2,W_3\}$ that is (a)~connected, (b)~disconnected, (c)~completely connected to $A$, and (d)~not completely connected to $A$. The black dots indicate individual spins and crosses highlight spins that form the support of $A$.}
    \label{fig:clusters}
\end{figure}

We associate with every cluster $\vec{W}$ a simple graph $G_\vec{W}$, the so-called cluster graph. The vertices of $G_\vec{W}$ correspond to the subsystems in $\vec{W}$, with repeated subsystems also appearing as repeated vertices. Two distinct vertices $X$ and $Y$ are connected by an edge if and only if the respective subsystems overlap, i.e., $X \cap Y \neq \emptyset$. We say that a cluster $\vec{W}$ is connected if and only if $G_\vec{W}$ is connected. We use the notation $\mathcal{G}_m$ for the set of connected clusters of size $m$, and $\mathcal{G} = \bigcup_{m \geq 1} \mathcal{G}_m$ for the set of all connected clusters. The following statement concerning the number of connected clusters is an essential ingredient of our algorithms.
\begin{lemma}[Proposition 3.6 of Haah \emph{et al.}~\cite{haah2021}]
    \label{le:clusters}
    Given a subsystem $X \in S$, the number of clusters in $\mathcal{G}_m$ that contain $X$ is bounded from above by $(e \mathfrak{d} )^m$. Moreover, there exists a deterministic classical algorithm with run time $\exp(\mathcal{O}(m \log \mathfrak{d}))$ that outputs a list of all such clusters.
\end{lemma}
\noindent The classical algorithm is given as Algorithm 1 in Section 3.4 of Haah \emph{et al.}~\cite{haah2021}.

The union $\vec{W} = \vec{V}_1 \cup \vec{V}_2$ of two clusters $\vec{V}_1$ and $\vec{V}_2$ is defined as the union of the multisets, adding all multiplicities such that $\mu_{\vec{W}}(X) = \mu_{\vec{V}_1}(X) + \mu_{\vec{V}_2}(X)$ for all $X \in S$. Another set of clusters of special interest is formed by the clusters connected to the support of a few-body operator $A$. We say that a cluster $\vec{W}$ is completely connected to $A$ if and only if the cluster graph $G_{\vec{W} \cup \{ \supp(A) \}}$ is connected, where we assume for simplicity that $A$ acts on a subsystem contained in the Hamiltonian such that $\supp(A) \in S$. We denote the set of such clusters of size $m$ by $\mathcal{G}_m^A$, and $\mathcal{G}^A = \bigcup_{m \geq 1} \mathcal{G}_m^A$. The different sets of clusters are illustrated in \figref{fig:clusters}.

Before proceeding, we introduce further notation related to clusters. It is sometimes convenient to assign a (nonunique) ordering to the subsystems in a cluster. For $\vec{W} \in \mathcal{C}_m$, we then write $\vec{W} = \{W_1, W_2, \ldots, W_m \}$. We make frequent use of the shorthands $\lambda^\vec{W} = \prod_{X \in S} \lambda_X^{\mu_\vec{W}(X)}$ and $\Vec{W}! = \prod_{X \in S} \mu_\vec{W}(X)!$. We often take derivatives with respect to the parameters of the Hamiltonian for all subsystems contained in a cluster $\vec{W}$. To this end, we define the cluster derivative $\mathcal{D}_\vec{W}$, which acts on any function of the Hamiltonian parameters $\lambda = \{\lambda_X: X \in S\}$ as
\begin{equation}
    \mathcal{D}_\vec{W} f(\lambda) = \left. \left[ \prod_{X \in S} \left( \frac{\partial}{\partial \lambda_X} \right)^{\mu_\vec{W}(X)} \right] f(\lambda) \right|_{\lambda = 0}.
\end{equation}
Here, the subscript $\lambda = 0$ means to set $\lambda_X = 0$ for all $X \in S$ after taking the derivatives. Hence, the cluster derivative isolates the contribution from the monomial $\lambda^\vec{W}$.

For any function $f(\lambda)$, we define its cluster expansion as the multivariate Taylor-series expansion in $\lambda$.
With the above notation, the cluster expansion can be concisely written as
\begin{equation}
    f(0) + \sum_{\vec{W} \in \mathcal{C}} \frac{\lambda^\vec{W}}{\vec{W}!} \mathcal{D}_\vec{W} f(\lambda).
\end{equation}
Our goal is to establish conditions under which the cluster expansion converges to $f(\lambda)$ for different functions of interest.

We illustrate the above concepts by an example in Appendix~\ref{sec:example}.

\subsection{Cluster partitions\label{sec:partitions}}
A partition $P$ of a cluster $\vec{W}$ is a multiset of clusters $\{ \vec{V}_1, \vec{V}_2, \ldots, \vec{V}_{|P|} \}$ such that $\vec{W} = \vec{V}_1 \cup \vec{V}_2 \cup \cdots \cup \vec{V}_{|P|}$. We are particularly interested in partitions where every element is a connected cluster. We refer to these as partitions of $\vec{W}$ into connected subclusters. The set of all such partitions is denoted by $\mathcal{P}_c(\vec{W})$.

We introduce several quantities characterizing cluster partitions. We use tildes to distinguish these from similar quantities describing clusters. The multiplicity $\tilde \mu_P(\clus{V})$ is defined as the number of times the cluster $\clus{V}$ appears in the partition $P$. The size $|P| = \sum_{\clus{V} \in \mathcal{G}} \tilde\mu_P(\clus{V})$ is the number of clusters in $P$, including their multiplicities. We also use the shorthand
\begin{equation}
    P! = \prod_{\clus{V} \in \mathcal{G}} \tilde\mu_P(\clus{V})!
\end{equation}

For every partition $P \in \mathcal{P}_c(\clus{W})$, we define a simple graph $\tilde G_P$, called the partition graph of $P$. The vertices of $\tilde G_P$ are the clusters in $P$. Two clusters $\clus{V}, \clus{V}' \in P$ are connected by an edge if and only if they overlap, that is, there exist subsystems $X \in \clus{V}$ and $Y \in \clus{V}'$ such that $X \cap Y \neq \emptyset$. Alternatively, we may obtain $\tilde G_P$ from $G_\clus{W}$ as follows. For every $\clus{V} \in P$, we merge the corresponding vertices in $G_\clus{W}$ into a single vertex and remove all loops. If any of the remaining edges are repeated, they are reduced to a single edge. \Figref{fig:graphs} shows an example of a partition graph and illustrates its connection to the cluster graph.

\begin{figure}[t]
    \centering
    \includegraphics[width=\columnwidth]{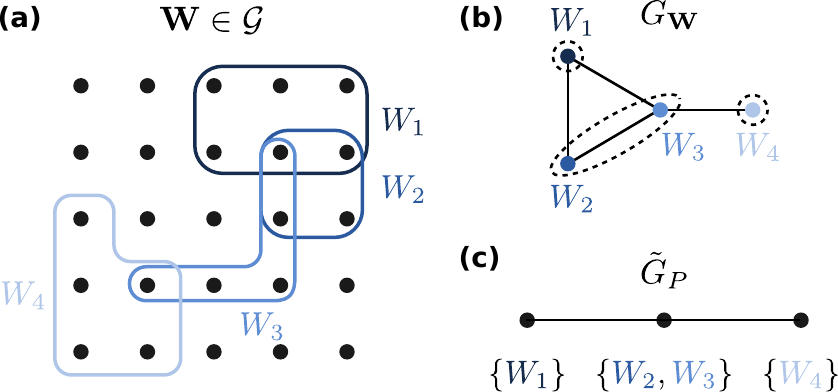}
    \caption{Illustration of cluster and partition graphs. (a)~A connected cluster $\vec{W} = \{W_1, W_2, W_3, W_4\}$ composed of four subsystems. The dots indicate individual spins. (b)~The corresponding cluster graph $G_\vec{W}$. The dashed outlines show the partition $P = \{\{ W_1\}, \{W_2, W_3\}, \{W_4\}\}$ of $\vec{W}$ into connected subclusters. (c)~The partition graph $\tilde{G}_P$ associated with $P$.}
    \label{fig:graphs}
\end{figure}

\section{Local observables}\label{sec:obs}
\subsection{Cluster expansion}

We consider the expectation value of an observable $A$ for an initial product state $\rho$ evolving under a local Hamiltonian $H$. After time $t$, we have
\begin{align}
    \langle A(t) \rangle = \tr \left( e^{i H t} A e^{- i H t} \rho \right) .
\end{align}
Due to the dependence of $H$ on the parameter set $\lambda$, we may think of $\langle A(t) \rangle$ as a function of $\lambda$. The corresponding cluster expansion is given by
\begin{equation}
    \label{eq:obs_cluster}
    f_A(t) = \langle A(0) \rangle + \sum_{m = 1}^\infty \sum_{\vec{W} \in \mathcal{C}_m} \frac{\lambda^\vec{W}}{\vec{W}!} \mathcal{D}_\vec{W} \langle A(t) \rangle.
\end{equation}
For $\vec{W} \in \mathcal{C}_m$, the cluster derivative can be explicitly evaluated as
\begin{align}
    \label{eq:commutator}
    \mathcal{D}_\vec{W} & \langle A(t) \rangle = \frac{(i t)^m}{m!} \\
    & \times \sum_{\sigma \in S_m} \tr \left( \comm{h_{W_{\sigma(1)}}} {\comm{h_{W_{\sigma(2)}}}{\cdots \comm{h_{W_{\sigma(m)}}}{A}}} \rho \right), \nonumber
\end{align}
where the sum runs over all permutations of the indices $\{1, 2, \ldots, m \}$.

For simplicity, we assume that the support of $A$ is contained in the set of subsystems $S$ of the local supports of the Hamiltonian, although this constraint can be readily relaxed. We establish the convergence of the cluster expansion for short times in two steps. First, we show that only clusters that are completely connected to $A$ contribute to the sum.
\begin{lemma}
    For any cluster $\vec{W} \notin \mathcal{G}^A$,
    \begin{equation}
        \mathcal{D}_\vec{W} \langle A(t) \rangle = 0.
    \end{equation}
\end{lemma}
\begin{proof}
    Consider $\vec{W} \notin \mathcal{G}_m^A$ for any integer $m > 0$. For every $\sigma \in S_m$, there exists a positive integer $k \leq m$ such that the intersection of $W_{\sigma(k)}$ with $W_{\sigma(k+1)} \cup W_{\sigma(k+2)} \cup \cdots \cup W_{\sigma(m)} \cup \supp(A)$ is empty. Then, $\comm{h_{W_{\sigma(k)}}} {\comm{h_{W_{\sigma(k+1)}}}{\cdots \comm{h_{W_{\sigma(m)}}}{A}}} = 0$, and the commutator in \eqref{eq:commutator} vanishes.
\end{proof}

Second, we need to bound the magnitude of the cluster derivative when it does not vanish. 
Note that there are at most $2^m$ in the nested commutators of \eqref{eq:commutator}. Repeated application of the triangle inequality then yields $\left| \tr \left( \comm{h_{W_{\sigma(1)}}} {\comm{h_{W_{\sigma(2)}}}{\cdots \comm{h_{W_{\sigma(m)}}}{A}}} \rho \right) \right| \leq 2^m \Vert A \Vert$. Hence, we find that for any $\vec{W} \in \mathcal{G}_m^A$,
\begin{equation}
    \left| \mathcal{D}_\vec{W} \langle A(t) \rangle \right| \leq (2 |t|)^m \Vert A \Vert.
\end{equation}

By combining these observations, we are able to prove convergence of the cluster expansion for short times:
\begin{proposition}\label{prop:shortTobs}
    Consider an operator $A$ for which $\supp(A) \in S$ and let $\vert t \vert < t^*= 1/(2 e \mathfrak{d})$. Then,
    \begin{align}\label{eq:shortTobs}
        &\left \vert \langle A(t) \rangle - \langle  A(0) \rangle - \sum_{m = 1}^M \sum_{\vec{W} \in \mathcal{G}_m^A} \frac{\lambda^\vec{W}}{\vec{W}!} \mathcal{D}_\vec{W} \langle A(t) \rangle \right \vert \nonumber\\
        & \hspace{4cm} \leq e \mathfrak{d} \Vert A \Vert \frac{(|t|/t^*)^{M+1}}{1-|t|/t^*} .
    \end{align}
\end{proposition}
\begin{proof}
    Consider the error of neglecting all clusters of size $m > M$ from the cluster expansion. Since $|\lambda_X| \leq 1$ for all $X \in S$, we can bound this error by
    \begin{equation}
        \label{eq:obs_error}
        \left \vert \sum_{m = M+1}^\infty \sum_{\vec{W} \in \mathcal{G}_m^A} \frac{\lambda^\vec{W}}{\vec{W}!} \mathcal{D}_\vec{W} \langle A(t) \rangle \right \vert \leq \sum_{m = M+1}^\infty (2 |t|)^m | \mathcal{G}_m^A | \Vert A \Vert.
    \end{equation}
    By assumption, $\supp(A) \in S$, which allows us to obtain every cluster in $\mathcal{G}_m^A$ by starting from a cluster $\vec{W} \in \mathcal{G}_{m+1}$ that contains $X = \supp(A)$ and reducing the multiplicity $\mu_\vec{W}(X)$ by 1. It follows from Lemma~\ref{le:clusters} that $\left| \mathcal{G}_m^A \right| \leq (e \mathfrak{d})^{m+1}$. Substituting this bound into \eqref{eq:obs_error} yields a geometric series, which converges when $|t| < t^*$. Convergence of this series implies the convergence of the cluster expansion such that $f_A(t) = \langle A(t) \rangle$ and leads to the error bound in the lemma.
\end{proof}

We highlight that Proposition~\ref{prop:shortTobs} holds for any quantum state $\rho$. The restriction to product states only becomes relevant when bounding the computational cost. In addition, the proposition is valid for complex values of $t$.


\subsection{Computation for short times}\label{sec:computationObs}

We now discuss the computational cost of estimating $\langle A(t) \rangle$ for a time $|t| < t^*$ up to additive error $\varepsilon \Vert A \Vert$. For simplicity, we only give the asymptotic scaling of the algorithm with $\varepsilon$ and $|t| / t^*$ and suppress the dependence on constant parameters such as the locality $k$ or the connectivity $\mathfrak{d}$ of the Hamiltonian. Moreover, we exclude issues of finite numerical precision, which have been addressed rigorously by Haah \emph{et al.}~\cite{haah2021}, from our considerations.

Our approximation algorithm computes the cluster expansion for all clusters up to size $M$, whose value is determined by $\varepsilon$ and $|t|/t^*$. For all $m \le M $, we proceed in two steps:
\begin{enumerate}[(i)]
    \item Enumerate all the connected clusters in $\mathcal{G}^A_m$.
    \item Compute and add the contributions of every cluster.
\end{enumerate}
Step (i) can be carried out by using the algorithm in Lemma~\ref{le:clusters} to first compute clusters of size $m+1$ containing $X = \supp(A)$ and by subsequently reducing the multiplicity of $X$ by one. The run time is $\exp (\mathcal{O}(m))$. The computational effort of step (ii) is dominated by the evaluation of the sum over nested commutators in \eqref{eq:commutator}. We show in Appendix \ref{app:nested} that  this can be done for a product state in time $\exp(\mathcal{O}(m))$ by suitably grouping the $m!$ terms of the sum.

Together, the two steps lead to the following run time of the approximation algorithm.
\begin{proposition}\label{prop:algoObsShort}
   Let $\rho$ be a product state and $\vert t \vert < t^*= 1/(2 e \mathfrak{d})$. There exists an algorithm that outputs an estimate $\hat f_A(t)$ with run time
    \begin{equation}\label{eq:run time1}
         \mathrm{poly}\left[ \left( \frac{1}{1 - |t|/t^*} \frac{1}{\varepsilon} \right)^{1/\log(t^*/|t|)} \right]
    \end{equation}
    such that $\left \vert \langle A(t) \rangle - \hat f_A(t) \right \vert \le \varepsilon \norm{A}$.
\end{proposition}
\begin{proof}
    According to Proposition~\ref{prop:shortTobs}, truncating the cluster expansion at order $M> \log \frac{e \mathfrak{d}}{ (1 - |t|/t^*) \varepsilon} / \log \frac{t^*}{|t|}$ leads to an error that is bounded from above by $\varepsilon \Vert A \Vert$. Following the above discussion of enumerating the clusters and computing their contributions, we see that the cluster expansion truncated at order $M$ can be evaluated in time $\exp(\mathcal{O}(M))$. Picking the smallest integer $M$ that guarantees the desired error bound yields \eqref{eq:run time1}.
\end{proof}

\subsection{Computation for arbitrary times}\label{sec:tcomputationObs}

The convergence result of Proposition~\ref{prop:shortTobs} is independent of the system size $n$. Hence, $\langle A(t) \rangle$ remains analytic in the thermodynamic limit for all complex values of $t$ that satisfy $|t| < t^*$.  Given any $t_0 \in \mathbb{R}$, we may write $\langle A(t) \rangle = \tr \left( e^{i H (t - t_0)} A e^{- i H (t - t_0)} \rho' \right)$, where $\rho' = e^{- i H t_0} \rho e^{i H t_0}$ is another quantum state. This shows that $\langle A(t) \rangle$ is analytic on a disk in the complex plane of radius $t^*$ around any point on the real axis. Stated differently, $\langle A(t) \rangle$ is analytic for all complex values of $t$ satisfying $| \im(t) | < t^*$. This analytic structure provides a strategy to compute $\langle A(t) \rangle$ for a product state $\rho$ for all $t \in \mathbb{R}$ and any system size $n$ by means of analytic continuation.

While there are many approaches to analytic continuation, we pick here a concrete method that employs a function $\phi(z)$ that maps a disk onto an elongated region along the real axis. For some $R > 1$ and $w > 0$, we assume that $\phi(z)$ satisfies the following three properties:
\begin{enumerate}[(i)]
    \item $\phi(z)$ is analytic on the closed disk $D_R = \{ z \in \mathbb{C} : \, |z| \leq R \}$,
    \item $\phi(0) = 0$ and $\phi(1) = 1$,
    \item $| \im ( \phi(z) ) | \leq w$ for all $z \in D_R$.
\end{enumerate}
We show below, using an explicit example, that such a function exists.

Next, we define $f(z) = \langle A(t \phi(z)) \rangle$, where $t \in \mathbb{R}$ is the time at which we want to evaluate $\langle A(t) \rangle$. It follows from property (ii) that $f(1) = \langle A(t) \rangle$. Because $\langle A(t) \rangle$ is analytic when $| \im(t) | < t^*$, properties (i) and (iii) together ensure that $f(z)$ is analytic on $D_R$, provided that $w|t| < t^*$. These relations are illustrated in \figref{fig:ac}.
\begin{figure}[t]
    \centering
    \includegraphics[width=\columnwidth]{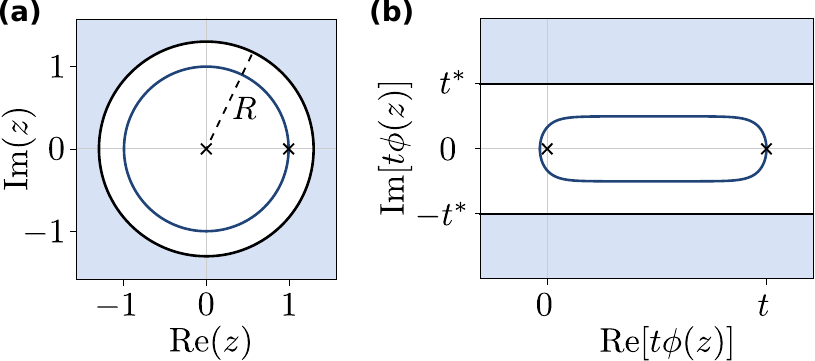}
    \caption{Illustration of the analytic continuation scheme. (a)~The function $\phi(z)$ is analytic in the domain $|z| \leq R$. The blue circle is the unit circle. (b)~The function $\langle A(t \phi(z)) \rangle$ is analytic for all $z$ such that $|\im[t \phi(z)]| < t^*$. Given $t \in \mathbb{R}$, $t \phi(z)$ maps the unit circle to the elongated shape outlined by the blue curve.}
    \label{fig:ac}
\end{figure}

Next, we compute the Taylor series of $f(z)$ at $z = 0$ up to order $M$. With $\langle A(t) \rangle = \sum_{l = 0}^\infty A_l t^l$ and $\phi(z) = \sum_{l = 0}^\infty \phi_l z^l$, we have
\begin{equation}
    \left. \frac{1}{k!} \frac{\di^k f(z)}{\di z^k} \right|_{z=0} = \sum_{l = 1}^k A_l t^l \sum_{\substack{m_1, \ldots, m_l \geq 1\\ m_1 + \cdots + m_l = k}} \phi_{m_1} \cdots \phi_{m_l},
\end{equation}
where we used the fact that $\phi_0 = 0$. We can obtain $A_l t^l$ from the cluster expansion as $A_l t^l = \sum_{\vec{W} \in \mathcal{G}_l} \lambda^\vec{W} \mathcal{D}_\vec{W} \langle A(t) \rangle / \vec{W}!$, which takes time $\exp(\mathcal{O}(l))$ to evaluate. The sum over the Taylor coefficients of $\phi(z)$ involves $\binom{l-1}{k-1} \leq 2^{l-1}$ terms. Assuming that the individual coefficients $\phi_l$ can be computed in time $\exp(\mathcal{O}(l))$, it thus takes time $\exp(\mathcal{O}(k))$ to compute the $k^\mathrm{th}$ Taylor coefficient of $f(z)$, and time $\exp(\mathcal{O}(M))$ to compute the full Taylor series up to order $M$.

To bound the truncation error of the Taylor series, we make use of the following standard result in complex analysis (see, e.g., Proposition 18 of reference~\cite{harrow2020}).
\begin{lemma}
    \label{le:taylor}
    We denote by $D_R$ the closed disk of radius $R$ centered at $z = 0$ in the complex plane, i.e., $D_R = \{z \in \mathbb{C} \vert \, \vert z \vert \leq R \}$. Let $f(z)$ be a complex function that is bounded by $|f(z)| \leq F$ and analytic for all $z \in D_R$. Given $\alpha < 1$, for all $z \in D_{\alpha R}$, the error of approximating $f(z)$ by a truncated Taylor series of order $M$ is bounded by
    \begin{equation}
        \varepsilon_M(z) = \left| f(z) - \sum_{m = 0}^M \frac{1}{m!} f^{(m)}(0) z^m \right| \leq \frac{\alpha^{M+1}}{1 - \alpha} F.
    \end{equation}
\end{lemma}

Combining this lemma with the above considerations for computing the Taylor series yields an algorithm to estimate $\langle A(t) \rangle$ for all $t \in \mathbb{R}$.
\begin{theorem}\label{th:algoObs}
    Given $t>0$, there is an algorithm that outputs the estimate $\hat f_A(t)$ with run time 
    \begin{equation}
        \mathrm{poly} \left[ \left( \frac{1}{\varepsilon} e^{\pi t / t^*} \right)^{\exp({\pi t /  t^*})} \right]
    \end{equation}
    such that for any product state $\rho$
    \begin{equation}
        \left \vert \langle A(t) \rangle - \hat f_A(t) \right \vert \le \varepsilon \norm{A}.
    \end{equation}
\end{theorem}

\begin{proof}
    We bound the truncation error of the Taylor series of $f(z)$ at $z = 1$ by letting $\alpha = 1/R$ in Lemma~\ref{le:taylor}. This yields
    \begin{equation}
        \varepsilon_M(1) \leq \frac{1}{R^M (R - 1)} \max_{z \in D_R} \left\vert f(z) \right\vert
    \end{equation}
   By the same argument we used to show that $\langle A(z) \rangle$ is analytic for all $|\im(z)| < t^*$, Proposition~\ref{prop:shortTobs} implies that 
    \begin{equation}
        \langle A (z) \rangle \leq \frac{e \mathfrak{d}}{1 - |\im(z)| / t^*} \Vert A \Vert
    \end{equation}
    for all $z$ satisfying $|\im(z)| < t^*$. Since $|\im(\phi(z))| \leq w$, assuming $w t / t^* < 1$, we have
    \begin{equation}
        \max_{z \in D_R} \left\vert f (z) \right\vert \leq \frac{e \mathfrak{d}}{1 - w t / t^*} \Vert A \Vert.
    \end{equation}
    Hence,
    \begin{equation}
        \label{eq:error}
        \varepsilon_M(1) \leq \frac{e \mathfrak{d}}{1 - w t / t^*} \frac{1}{R^M (R - 1)} \Vert A \Vert.
    \end{equation}
    The parameters $R$ and $w$ cannot be chosen entirely at will owing to the constraints on $\phi(z)$. We do not attempt to address this issue in general but instead consider the concrete example $\phi(z) = \log(1 - z/R')/\log(1 - 1/R')$ with $R' > R > 1$, inspired by section Lemma 2.2.3 of Reference~\cite{barvinok2016}. This function clearly satisfies the requirements (i)--(ii). For requirement (iii), we assume that the branch of the logarithm is chosen such that $| \im(\phi(z))| \leq -\pi / [2 \log(1 - 1/R')]$. We therefore set $w = - \pi / [2 \log(1 - 1/R')]$. Recalling that $w < t^* / t$, we separately let $w t / t^* = \eta$ for some $\eta < 1$. We can always choose $R < R'$ such that $R'^m (R' - 1) = 2 R^M (R - 1)$, allowing us to replace $R$ by $R'$ in \eqref{eq:error} at the cost of a factor of $2$. For the particular choice of $\phi(z)$, we thus obtain
    \begin{equation}
        \varepsilon_M(1) \leq \frac{2 e \mathfrak{d}}{1 - \eta} \left(1 - e^{- \pi t / 2 \eta t^*} \right)^M \left( e^{\pi t / 2 \eta t^*} - 1 \right) \Vert A \Vert.
    \end{equation}
    It follows from this expression that truncating at order $M > e^{\pi t / 2 \eta t^*} \log\left( \frac{2 e \mathfrak{d}}{1 - \eta} \frac{1}{\varepsilon} e^{\pi t/ 2 \eta t^*} \right)$ guarantees an error  $\varepsilon_M(1) \leq \varepsilon \Vert A \Vert$. Thus, our output $\hat f_A(t)$ is the Taylor series of $f (z)$ at $z=1$ up to the smallest integer satisfying the lower bound on $M$. Since the computational cost is exponential in $M$, the theorem follows by setting $\eta=1/2$.
\end{proof}

The above approach yields an algorithm to estimate $\langle A(t) \rangle$ with a computational cost that scales polynomially with $1/\varepsilon$ for any fixed real time $t$. The cost, however, has a doubly exponential dependence on $t / t^*$, which renders this approach unsuitable for practical computations at long times. The chain of disks, another common method of analytic continuation, also yields a doubly exponential dependence of the computational cost on the time $t$~\cite{trefethen2020}. We conjecture that this scaling is an unavoidable consequence of analytic continuation. More efficient continuation algorithms may be available if the analytic domain of $\langle A(t) \rangle$ extends along the imaginary direction beyond a constant. This is, however, not possible in general because there exist local observables and Hamiltonians for which $\langle A(t) \rangle$ becomes nonanalytic in the thermodynamic limit at a constant imaginary time~\cite{bouch2015}. This also indicates that our convergence result is optimal up to an improvement of $t^*$ by a constant factor. We remark that the above procedure can be adapted to not only yield expectation values with product initial states, but also an $M$-local approximation of the operator $A(t)$.


\subsection{Comparison with the Lieb--Robinson bound}
The Lieb--Robinson bound \cite{Lieb1972} offers an alternative method for computing the time evolution of local operators. It implies that 
\begin{equation}
   \left \vert \langle A (t) \rangle -   \langle e^{iH_l t} A e^{-iH_l t} \rangle  \right \vert \le C e^{(v t - l)/\xi} \Vert A \Vert,
\end{equation}
where $C$ and $\xi$ are constants, and $v$ is the Lieb-Robinson velocity. The Hamiltonian $H_l=\sum_{X \in S: \, \text{dist}(X,A)\le l} \lambda_X h_X$ contains all local terms within the graph distance $\text{dist}(X,A)$ between the operators $X$ and $A$ on the interaction graph $G$.

To compute $\langle A(t) \rangle$ to within additive error $\varepsilon \Vert A \Vert$, it suffices to compute $\langle e^{-iH_l t} A e^{iH_l t} \rangle $ on a region of radius $l = v t + \xi \log (C / \varepsilon)$. In a lattice in $D$ dimensions, the computational cost of performing exact diagonalization on this region is exponential in $\left[ v t + \xi \log(C/\varepsilon) \right]^D = \xi^D \log^D \left[ e^{vt/\xi} \left( C / \varepsilon \right) \right]$. For $D > 1$, this yields an algorithm whose cost is superpolynomial in both $e^{v t / \xi}$ and $1/\varepsilon$, as opposed to the polynomial dependence on $1/\varepsilon$ in Theorem~\ref{th:algoObs}. The difference is even more marked in expander graphs, for which the number of sites at a distance $l$ grows as $\exp(\mathcal{O}(l))$. The Lieb-Robinson method has a run time that is exponential in $\mathrm{poly}(1/\varepsilon)$ and doubly exponential in $t$ for such graphs.


\section{Loschmidt echo}\label{sec:loschmidt}

\subsection{Cluster expansion}\label{sec:clusterexp}

We now focus on the Loschmidt echo 
\begin{equation}
   L(t) \equiv \tr \left(  e^{-i Ht} \rho \right),
\end{equation}
where $\rho$ is a product state on the $n$ qubits. This is an important quantity in the study of dynamics of quantum systems. It appears in diverse contexts such as quantum chaos~\cite{gorin2006}, as the characteristic function of the local density of states \cite{Emerson_2004}, and in algorithms for quantum simulation~\cite{Lu_2021}. As we discuss below, it is also a key quantity in the description of dynamical phase transitions and other relevant phenomena.

We consider the logarithm $\log L(t)$, whose multivariate Taylor expansion can be written in terms of cluster derivatives as
\begin{equation}\label{eq:expansion}
    f_L(t) = \sum_{m=1}^{\infty} \sum_{\vec{W} \in \mathcal{C}_m} \frac{\lambda^\vec{W}}{\vec{W}!} \mathcal{D}_\vec{W} \log L(t).
\end{equation}
Our goal is to establish sufficient criteria for the convergence of this expansion. Working with the logarithm of $L(t)$ greatly reduces the number of clusters involved since only connected ones contribute:
\begin{lemma}\label{le:connected}
    If $\rho$ is a product state, then for any disconnected cluster $\vec{W}$,
    \begin{equation}
        \mathcal{D}_\vec{W} \log L(t) = 0.
    \end{equation}
\end{lemma}
\begin{proof}
    Since $\vec{W} \notin \mathcal{G}$, there exists a decomposition $\vec{W} = \vec{W}_A \cup \vec{W}_B$ such that the cluster graphs $G_{\vec{W}_A}$ and $G_{\vec{W}_B}$ are disconnected components of $G_\vec{W}$. We define $H_A = \sum_{X \in \vec{W}_A} \lambda_X h_X$ and $H_B = \sum_{X \in \vec{W}_B} \lambda_X h_X$, where each subsystem is included at most once in the sum, even if it appears with higher multiplicity in the cluster. Clearly, $\supp (H_A) \cap \supp ( H_B ) = \emptyset$, which implies $\comm{H_A}{H_B} = 0$ and
    \begin{equation}
        \mathcal{D}_{\vec{W}} \log L(t) = \mathcal{D}_\vec{W} \log \tr \left( e^{- i H_A t} e^{- i H_B t} \rho \right)
    \end{equation}
    For $\rho$ a product state, we further have
    \begin{equation}
       \tr \left( e^{- i H_A t} e^{- i H_B t} \rho \right) =\tr \left( e^{- i H_A t} \rho \right) \tr \left( e^{- i H_B t} \rho \right)
    \end{equation}
    and thus
    \begin{align}     
        &\mathcal{D}_{\vec{W}} \log L(t) \nonumber \\
        & \quad = \mathcal{D}_\vec{W} \log \tr \left( e^{- i H_A t} \rho \right) + \mathcal{D}_\vec{W} \log \tr \left(  e^{- i H_B t} \rho \right).
    \end{align}
    The first term vanishes because $H_A$ is independent of $\lambda_X$ for any $X \in \vec{W}_B$. Similarly, the second term is independent of  $\lambda_X$ for $X \in \vec{W}_A$.
\end{proof}

Hence, we can restrict the sum in \eqref{eq:expansion} to connected clusters. 

We next bound the cluster derivative using the bounded connectivity of the cluster and partition graphs. Many properties of a graph $G$ are captured by the Tutte polynomial $T_G(x, y)$ (see, e.g., \cite{biggs1993algebraic}). For instance, $T_G(1, 1)$ is the number of spanning trees (or spanning forests if the graph is not connected), $T_G(2, 1)$ is the number of forests and $T_G(2, 2)$ is $2^{\vert E \vert}$, where $\vert E \vert$ is the number of edges. 

Using the Tutte polynomial, the cluster derivative of $\log L(t)$ can be concisely expressed as in the following lemma.
\begin{lemma}\label{le:cldev}
    For any $\vec{W} \in \mathcal{G}_m$, the cluster derivative of $\log L(t)$ can be written as
    \begin{align}
        \mathcal{D}_\vec{W} & \log  L(t)  = (-i t)^m \\
        & \times \sum_{P \in \mathcal{P}_c(\vec{W})} (-1)^{|P|-1} N_P(\vec{W}) T_{\tilde{G}_P}(1, 0) \prod_{\vec{V} \in P} \langle h^\vec{V} \rangle_s, \nonumber
    \end{align}
    where we introduced the symmetrized expectation value
    \begin{equation}
        \langle h^\vec{V} \rangle_s =  \frac{1}{|\vec{V}|!} \sum_{\sigma \in S_{|\vec{V}|}} \tr \left( h_{V_{\sigma(1)}} h_{V_{\sigma(2)}} \cdots h_{V_{\sigma(|\vec{V}|)}} \rho \right).
    \end{equation}
    and the combinatorial factor
    \begin{equation}
        N_P(\vec{W}) = \frac{\vec{W}!}{P! \prod_{\vec{V} \in P} \vec{V}!}.
    \end{equation}
\end{lemma}
We prove this lemma in Appendix~\ref{app:cldev}. A similar statement has been reported by Mann and Helmuth~\cite{mann2021}, with Helmuth, Perkins, and Regts pointing out the relation to the Tutte polynomial~\cite{helmuth2020}. 

We make use of Lemma~\ref{le:cldev} to derive the following upper bound on the cluster derivative.
\begin{proposition}\label{le:cldev_bound}
Let $\vec{W} \in \mathcal{G}_m$. Then,
\begin{equation}\label{eq:cldev_bound2}
   \left\vert \frac{1}{\vec{W}!} \mathcal{D}_{\vec{W}} \log L(t) \right\vert \le [2 e (\mathfrak{d}+1) |t|]^{m} .
\end{equation}
\end{proposition}
We defer the proof to Appendix \ref{app:cldev_bound}. In rough terms, it proceeds by observing that $|\langle h^\vec{V} \rangle_s | \leq 1$ and $0 \leq T_{\tilde G_P}(1, 0) \leq T_{\tilde G_P}(1, 1)$, where $T_{\tilde G_P}(1, 1)$ is equal to the number of spanning trees of $\tilde G_P$. The sum of these trees over all partitions can then be bounded by the number of spanning trees of the original cluster graph $G_{\vec{W}}$. This number is smaller than an exponential in $m$ times $\clus{W}!$, yielding \eqref{eq:cldev_bound2}.

Having bounded each term in the cluster expansion, it remains to bound the number of terms, i.e.~the number of connected clusters of size $m$. This is done with Lemma~\ref{le:clusters} and the fact that there are $\vert S \vert$ subsystems on the lattice such that the total number is bounded by $\vert S \vert (e \mathfrak{d})^m$.

The main result of this section is the following theorem.
\begin{theorem}\label{th:taylor}
    The logarithm of the Loschmidt echo, $\log L(t)$, is analytic for $\vert t \vert < t_L^* = 1/[2 e^2 \mathfrak{d} (\mathfrak{d}+1)]$ and the truncation error of the cluster expansion can be bounded by
    \begin{align}\label{eq:taylorM}
        & \left\vert \log L(t) - \sum_{m=1}^M \sum_{\vec{W} \in \mathcal{G}_m} \frac{\lambda^\vec{W}}{\vec{W}!} \mathcal{D}_\vec{W} \log L(t) \right\vert \\
        & \hspace{5cm} \le \vert S \vert  \frac{(\vert t\vert/t_L^*)^{M+1}}{1- \vert t\vert/t_L^*}. \nonumber
    \end{align}
\end{theorem}
\begin{proof}
    Lemma~\ref{le:connected} and Propositions~\ref{le:cldev_bound}, together with the bound on the number of clusters, imply that
    \begin{align}
       \left| \sum_{\vec{W} \in \mathcal{G}_m} \frac{\lambda^\vec{W}}{\vec{W}!} \mathcal{D}_\vec{W} \log L(t) \right| \leq \vert S \vert  [2 e^2 \mathfrak{d} (\mathfrak{d}+1) |t|]^{m}.
    \end{align}
    The result then follows by considering the weight of the terms of the Taylor expansion with $m > M$.
\end{proof}

The overall argument of this section mirrors the steps of previous results on Gibbs states~\cite{haah2021,kuwahara2020_gaussian}.  Our new technical contributions are the expression for the cluster derivative in Lemma~\ref{le:cldev} and the bound in Proposition~\ref{le:cldev_bound}, for which we use a novel proof strategy based on counting trees. We note that convergence results similar to Theorem~\ref{th:taylor} can also be proved using the general framework of abstract polymer models~\cite{KoteckyPreiss,dobrushin1996estimates,mann2021}.


\subsection{Computation of the Loschmidt echo}\label{sec:classalgo}

The above Taylor approximation allows for an efficient classical estimation of the Loschmidt echo for short times. Theorem \ref{th:taylor} guarantees that to approximate $\log L(t)$ it suffices to calculate the terms in the series up to some order $M$. Recall that the $m$$^\mathrm{th}$ order term is given by
\begin{equation}\label{eq:cmcluster}
    \sum_{\vec{W} \in \mathcal{G}_m} \frac{\lambda^\vec{W}}{\vec{W}!} \mathcal{D}_\vec{W} \log L(t).
\end{equation}

\noindent As in Sec.~\ref{sec:computationObs}, we need two steps to calculate these: (i) enumerate all the connected clusters in $\mathcal{G}_m$ and (ii) compute and sum the cluster derivative for every connected cluster according to \eqref{eq:cmcluster}. Lemma~\ref{le:clusters} addresses the first step. For the second one, we need to bound the cost of computing cluster derivatives. Related bounds have previously been stated in several works~\cite{kuwahara2020_clustering,haah2021,mann2021}, where the computation hase proceeded either by directly differentiating $\log L(t)$ or by summing the terms of an expansion similar to that in Lemma~\ref{le:cldev}. Here we pursue the latter approach, stating the computational cost in the following proposition. As in Sec.~\ref{sec:obs}, we ignore complications arising due to finite numerical precision.

\begin{proposition}\label{le:algo}
    There exists a deterministic algorithm that outputs $\mathcal{D}_\vec{W} \log L(t)$ for $\vec{W} \in \mathcal{G}_m$ with running time $\exp(\mathcal{O}(m))$.
\end{proposition}
\begin{proof}
    The algorithm evaluates the expression in Lemma~\ref{le:cldev}. There are three nontrivial contributions to the run time:
    \begin{enumerate}[(i)]
        \item Enumerating the partitions of $\vec{W}$ into connected subclusters. This takes time $\exp(\mathcal{O}(m))$ by the following algorithm. Assign to each subsystem in $\vec{W}$ a unique label from the set $\{1, 2, \ldots, m\}$. List all compositions of $m$, i.e., ordered tuples of positive integers $(n_1, n_2, \ldots, n_l)$ such that their sum equals $m$. There are $2^{m-1}$ distinct compositions, which can be enumerated in time $\exp(\mathcal{O}(m))$. For each composition, find all connected clusters $\vec{V}_1$ of size $n_1$ that are contained in $\vec{W}$ and include the subsystem labeled $1$. By Lemma~\ref{le:clusters}, this step can be carried out in time $\exp(\mathcal{O}(n_1))$ by enumerating all clusters connected to subsystem $1$ and removing the ones that are not contained in $\vec{W}$. Next, for each $\vec{V}_1$, find all connected clusters $\vec{V}_2$ of size $n_2$ that are contained in $\vec{W} \setminus \vec{V}_1$ and include the subsystem with the smallest label remaining in $\vec{W} \setminus \vec{V}_1$. This step takes a computational time $\exp(\mathcal{O}(n_2))$. We iterate this procedure, removing the new cluster $\vec{V}_i$ from the original cluster in every step until $\vec{W} = \bigcup_{i = 1}^l \vec{V}_i$. The procedure takes time $\exp(\mathcal{O}(n_1))\exp(\mathcal{O}(n_2)) \cdots \exp(\mathcal{O}(n_l)) = \exp(\mathcal{O}(m))$. The above steps produce a list of length $\exp(\mathcal{O}(m))$, which includes all desired partitions. Duplicates may appear, although these can be removed in time $\exp(\mathcal{O}(m))$.
        \item Computing the Tutte polynomial $T_{\tilde G_P}(1, 0)$, where $\tilde G_P$ has $|P| \leq m$ vertices. This can be done in time $\exp(\mathcal{O}(|P|))$ using the algorithm by Björklund \emph{et al.}~\cite{bjorklund2008}.
        \item Computing the symmetrized expectation value $\langle h^\vec{V} \rangle_s$ for $\vec{V} \in \mathcal{G}_l$ with $l \leq m$. The approach in Appendix~\ref{app:nested}---\eqref{eq:permutation} in particular---can be be adapted to carry out this computation in time $\exp( \mathcal{O}(l))$.
    \end{enumerate}
\end{proof}

With these ingredients, the main result is as follows. 
\begin{theorem}\label{th:computation}
For times $|t|<t_L^* = 1/[2 e^2 \mathfrak{d} (\mathfrak{d}+1)]$, there exists a classical algorithm with run time 
\begin{equation}\label{eq:run time}
    |S| \times \mathrm{poly}\left[ \left( \frac{|S|}{1 - |t|/t_L^*} \frac{1}{\varepsilon} \right)^{1/\log(t_L^*/|t|)} \right]
\end{equation}
that outputs $\hat f_L(t)$ such that 
$
    \vert \log L(t) - \hat f_L(t) \vert \le \varepsilon
$.
\end{theorem}
\begin{proof}
    Theorem~\ref{th:taylor} implies that the truncation error of the cluster expansion is smaller than $\varepsilon$ when keeping terms up to $M > \log \frac{\vert S \vert}{(1-|t|/t_L^*) \varepsilon} /\log \frac{t_L^*}{t}$. With Lemma~\ref{le:clusters} and Proposition~\ref{le:algo}, this determines the run time in \eqref{eq:run time}.
\end{proof}

For a fixed $t < t^*_L$, the run time is polynomial in $1/\varepsilon$ as well as in the number of terms $\vert S \vert$, which is proportional to the system size $n$. The output approximates the Loschmidt echo by $\exp[\hat f_L(t)]$ with a multiplicative error
\begin{equation}
 e^{-\varepsilon} \left| e^{ \hat{f}_L(t) } \right|  \le | L(t) | \le  e^{\varepsilon} \left| e^{ \hat{f}_L(t) } \right| .
 \label{eq:multiplicative}
\end{equation}

Unlike in the case of local observables, it is in general not possible to analytically continue the Loschmidt echo beyond $t_L^*$ because the zeros of $L(t)$, and thus nonanalyticities of $\log L(t)$, may be located anywhere in the complex plane, including the real axis.


\subsection{Generalized Loschmidt echo}

We now show that similar results hold for a Loschmidt echo with multiple Hamiltonians, defined as
\begin{equation}
 L(t_1, t_2, \ldots, t_K)=L(\{t_l\})= \tr \left( \prod_{l=1}^K e^{-i H^{(l)} t_l} \rho \right).
\end{equation}
We assume that each of the $K$ Hamiltonians $\{H^{(l)} \}$ satisfies the same conditions as in Sec. \ref{sec:setup}, so that they can all be written as
\begin{equation}
    H^{(l)}=\sum_{X \in S} \lambda^{(l)}_X h_X^{(l)}.
\end{equation}

We define the set of labeled subsystems $S^K = \{(X, l) : \, X \in S, \, l \in \{1, 2, \ldots K \} \}$, where the additional index keeps track of the Hamiltonian. The corresponding interaction graph $G^K$ may be viewed as $K$ copies of the original interaction graph $G$, where vertices are connected if the subsystems overlap, independent of which copy they belong to. Note that if the maximum degree of the original interaction graph $G$ was $\mathfrak{d}$, then the maximum degree of $G^K$ is $K (\mathfrak{d} + 1) - 1$.

A cluster $\vec{W}$ is now defined as a multiset of elements from $S^K$, with the set of clusters of size $m$ denoted by $\mathcal{C}_m^K$ and the set of connected clusters by $\mathcal{G}_m^K$. The multiplicities of a subsystem and Hamiltonian label in a cluster are denoted by $\mu_\vec{W}((X, l))$. The cluster graph $G_\vec{W}^K$ is again constructed by connecting subsystems that overlap, independent of the Hamiltonian with which they are associated.

With this notation, the logarithm of the Loschmidt echo permits the cluster expansion
\begin{equation}\label{eq:taylorL}
    f_L(\{t_l\}) = \sum_{m=1}^{\infty} \sum_{\vec{W}  \in \mathcal{G}^K_m} \frac{\lambda^{\vec{W}}}{\vec{W}!} \mathcal{D}_{\vec{W}} \log L(\{t_l\}) .
\end{equation}
Here, we introduced natural generalizations of our shorthands, $\lambda^\vec{W} = \prod_{(X, l) \in S^K} \left( \lambda^{(l)}_X \right)^{\mu_\vec{W}((X, l))}$, $\Vec{W}! = \prod_{(X, l) \in S^K} \mu_\vec{W}((X, l))!$, and the cluster derivative
\begin{equation}
    \mathcal{D}_{\vec{W}} f(\{t_l\}) = \left. \left[ \prod_{(X, l) \in S^K} \left( \frac{\partial}{\partial \lambda^{(l)}_X} \right)^{\mu_{\vec{W}}((X,l))} \right] f(\{t_l\}) \right|_{\lambda = 0}.
\end{equation}

As we show in Appendix~\ref{app:generalizedL2}, the results of Sec. \ref{sec:clusterexp} carry over directly, as long as we take into account the increased maximum degree of the interaction graph $G^K$. In particular, this means that for $\vec{W} \in \mathcal{G}_m^K$,
\begin{equation}
    \label{eq:cld}
    \left| \frac{1}{\vec{W}!} \mathcal{D}_\vec{W} \log L(\{t_l\}) \right| \leq [2 e K (\mathfrak{d} + 1)]^m \prod_{l = 1}^K |t_l|^{m_l},
\end{equation}
where $m_l$ is the number of terms in $\vec{W}$ associated with $H^{(l)}$. Note the additional factor of $K$, which effectively reduces the threshold time by $1/K$.

The analysis of the computational cost is similar to that in Section~\ref{sec:classalgo}, as detailed in Appendix~\ref{app:generalizedL2}. Hence, we obtain an analog of Theorems~\ref{th:taylor} and \ref{th:computation}.
\begin{theorem}\label{th:taylorL}
    Let $\tau = K \sum_{l=1}^K \vert t_l \vert / t_L^*$, where $t_L^* = 1/[2 e^2 \mathfrak{d} (\mathfrak{d} + 1)]$. The logarithm of the generalized Loschmidt echo, $\log L(\{t_l\})$, is analytic for $\tau < 1$ and its Taylor series converges as
    \begin{align}\label{eq:taylorML}
        \left \vert \log L(\{t_l\}) - \sum_{m=1}^M \sum_{\vec{W}  \in \mathcal{G}_m^K} \prod_{l=1}^K \frac{\lambda^{\vec{W} }}{\vec{W}!} \mathcal{D}_{\vec{W}} \log L(\{t_l\}) \right \vert \nonumber \\
        \le \vert S \vert  \frac{\tau^{M+1}}{1 - \tau}.
    \end{align}
    Moreover, there exists a classical algorithm with run time 
    \begin{equation}
        \label{eq:run timeL}
        |S| \times \mathrm{poly}\left[ \left( \frac{|S|}{1 - \tau} \frac{1}{\varepsilon} \right)^{K/\log(1/\tau)} \right]
    \end{equation}
    that outputs $\hat f_L(\{t_l\})$ such that 
    $
        \vert \log L(\{t_l\}) - \hat f_L(\{t_l\}) \vert \le \varepsilon
    $.
\end{theorem}
We highlight again that the computational cost scales polynomially with the number of Hamiltonian terms $\vert S \vert$ and the inverse error $1/\varepsilon$. The number of Hamiltonians $K$ enters in the shortened threshold time $t_L^* / K$ and in the exponent in \eqref{eq:run timeL} as a computational overhead.


\section{Further Implications}\label{sec:implications}
Besides leading to efficient classical algorithms, the convergence of the cluster expansion has important physical consequences, which we discuss below.

\subsection{Concentration bounds}

Given a quantum state $\rho$, the outcome of a measurement of an observable $A$ is a random variable. Assuming a projective measurement onto the eigenspaces of $A$, the probability of outcome $x$ is given by $\mathrm{Pr}(x) = \tr[ \Pi(x) \rho]$, where $\Pi(x)$ is the projector onto the eigenspace of $A$ with eigenvalue $x$. A concentration bound is an upper bound on the probability $\mathrm{Pr}[|x - \tr(\rho A)| \geq \delta]$, i.e., the probability that $x$ deviates from its mean $\tr(\rho A)$ by more than $\delta$. We will focus on the Hamiltonian $H = \sum_X \lambda_X h_X$ as the observable. The energy distribution of a state plays an important role in thermalization and equilibration \cite{wilming2018equilibration,kuwahara2020eigenstate} and its concentration properties place limitations on the performance of variational quantum algorithms~\cite{De_Palma_2023,AnshuConc2022}.

As a warm-up example, suppose the Hamiltonian terms $h_X$ act on distinct single sites and $\rho$ is a product state. Then, the measurement outcome of $H$ is equal to the sum of the measurement outcomes of $h_X$, which are independent random variables. The Chernoff--Hoeffding bound for independent random variables implies that
\begin{equation}
    \mathrm{Pr}[|x - \tr(\rho H)| \geq \delta] \leq 2 \exp\left( -\frac{\delta^2}{2 \sum_X \Vert h_X \Vert^2} \right).
    \label{eq:chernoff}
\end{equation}
We note that the denominator in the exponent is proportional to the system size $n$ such that deviations from the mean energy of order $\sqrt{n}$ are strongly suppressed.

This simple argument fails when the terms $h_X$ in the Hamiltonian overlap or when $\rho$ is not a product state because the outcomes of the measurements $h_X$ are no longer independent. Nevertheless, similar bounds hold for local Hamiltonians and sufficiently weakly correlated states. A number of proof techniques have been employed to establish such results  \cite{Kuwahara_2016,Anshu_2016,De_Palma_2022}, including the cluster expansion~\cite{kuwahara2020_clustering} (see also, e.g., Ref.~\cite{Neto_n__2004} for results on large deviations).

To illustrate the method, we use the cluster expansion to give a concise proof of a concentration bound for the energy of product states, reproducing the main result of \cite{Kuwahara_2016}. It follows from a standard argument (see, e.g., Corollary~1 in \cite{kuwahara2020_clustering}) that the probability of obtaining a measurement outcome $x$ that is greater than the average energy by at least $\delta$ is bounded according to
\begin{align}
    \text{Pr} \left[x - \tr(\rho H) \ge \delta \right] \le e^{- \delta \lc} \tr \left[ \rho \, e^{\lc (H-\tr(\rho H)) } \right].
    \label{eq:concentration}
\end{align}
Next, we apply Theorem~\ref{th:taylor} with $t=i \lc$, where $0 \leq \lc < t_L^*$. Taking $M=1$, the theorem implies that
\begin{equation}
     \left| \log \tr(\rho e^{\lc H}) - \lc \, \tr(\rho H) \right| \le \vert S \vert \frac{(\lc/t^*_L)^2}{1-\lc/{t^*_L}}.
\end{equation}
Substituting into \eqref{eq:concentration} yields
\begin{equation}
    \text{Pr} \left[x - \tr(\rho H) \ge \delta \right] \le \exp \left[ - \delta \lc + \vert S \vert \frac{(\lc/t^*_L)^2}{1-\lc/{t^*_L}} \right]. 
\end{equation}

Choosing $\lc= {\delta (t^*_L)^2}/{4 \vert S \vert}$ (which is always possible since $\delta / |S| \le 1$ and $t^*_L < 1$), and applying the same bound for  $\text{Pr} (x - \tr[\rho H] \le \delta)$, we obtain
\begin{equation}
     \text{Pr} [\vert x - \tr(\rho H) \vert \ge \delta] \le 2 e^{-{(\delta t^*_L)^2}/{8 \vert S \vert}}.
\end{equation}
We observe that this bound has the same dependence on $\delta$ and the system size as the bound for the special case in \eqref{eq:chernoff}.

Theorem~\ref{th:taylorL} enables us to extend the bound beyond product states. For short times, we can consider
\begin{align}
    \Big| \log \tr &\left( e^{- it H^{(1)}} \rho e^{it H^{(1)}} e^{\lc H^{(2)}} \right)  \\
    & \qquad - \lc \, \tr\left( e^{- it H^{(1)}} \rho e^{it H^{(1)}} H^{(2)} \right) \Big|, \nonumber
\end{align}
from which we obtain the following concentration bound.
\begin{corollary} \label{co:conc}
Let $H^{(1)}$ and $H^{(2)}$ be local Hamiltonians, let $\rho$ be a product state, and let $\vert t \vert \le t^*_L /7$. The probability that a projective measurement of the state $e^{-it H^{(1)}}\rho e^{it H^{(1)}}$ onto the eigenbasis of $H^{(2)}$ yields a value $x$ that deviates from the expectation value $\tr [ e^{-it H^{(1)}}\rho e^{it H^{(1)}} H^{(2)} ]$ by at least $\delta$ is bounded from above by
\begin{align}
    \mathrm{Pr}   \left[ \left\vert x - \tr\left( e^{-it H^{(1)}}\rho e^{it H^{(1)}} H^{(2)} \right) \right\vert \ge \delta \right]  \le 2 e^{\frac{- {(\delta t^*_L)^2}}{{(250)^2 \vert S \vert}}}.
\end{align}

\end{corollary}
The proof is shown in Appendix \ref{app:concentration}. The corollary shows that after evolving a product state for a short time under a local Hamiltonian, the energy distribution with respect to a (possibly different) local Hamiltonian is concentrated around the mean. Previous results along these lines cover many relevant cases but not time-evolved product states~\cite{Kuwahara_2016,Anshu_2016,kuwahara2020_gaussian,AnshuConc2022,De_Palma_2022}. Since Theorem \ref{th:taylorL} works for any number of Hamiltonians, this corollary can also be extended to states of the form $ e^{-it H^{(1)}} \cdots e^{-it H^{(K)}}\rho e^{it H^{(K)}} \cdots e^{-it H^{(1)}}$, which often feature in variational quantum algorithms, for a correspondingly shorter threshold time.

\subsection{Dynamical phase transitions}

 Dynamical phase transitions (DPTs) may be viewed as a real-time analog of thermal phase transitions~\cite{heyl2015}. The cluster expansion naturally constrains the time at which they can appear. Let us consider an infinite sequence of local Hamiltonians $H_n$ on $n$ particles under the assumptions of Section \ref{sec:setup} and product states $\ket{\Phi}=\otimes_{i=1}^n \ket{\phi}_i$, such that $g_{n}(t)\equiv \log \bra{\Phi} e^{-i tH_{n}} \ket{\Phi}$. A DPT occurs when the following function is nonanalytic~\cite{heyl2018}:
\begin{equation}
    G(t)=\lim_{n \rightarrow \infty} \frac{g_n(t)}{n}.   
\end{equation}

The following result is a consequence of Theorem \ref{th:taylor}.
\begin{corollary}\label{co:DPTs}
$G(t)$ is analytic for $t < t^*_L$, and thus DPTs can occur only at later times.
\end{corollary}

The proof is shown in Appendix \ref{app:DPTs}. The time scale $t^*_L$ is hence a universal lower bound on the time at which dynamical phase transitions occur. Nonanalyticities can appear in the logarithm of the Loschmidt echo at times $t \sim \mathcal{O}(1)$ as can be seen from the simple case of noninteracting spins. This has also been demonstrated analytically for particular interacting models in one dimension~\cite{piroli2018}. This shows that $t^*_L$ in Theorem \ref{th:taylor} can be increased at most by a constant factor. Note also that Theorem \ref{th:taylorL} allows us to extend this corollary to some Floquet systems. The absence of dynamical phase transitions at short times is analogous to the fact that thermal phase transitions can only occur above some threshold inverse temperature $\beta^*$ that depends on the details of the system.


\subsection{Quantum speed limits}

A quantum speed limit (QSL) is a bound on the time $t_\mathrm{QSL}$ that it takes for a state $\Phi$ evolving under a Hamiltonian $H$ to become orthogonal to itself. Formally,
\begin{equation}
    t_\mathrm{QSL}= \min \{t : \bra{\Phi} e^{-it H} \ket{\Phi}=0\}.
\end{equation}
The best-known general limits are the  Mandelstam-Tamm and the Margolus-Levitin bounds \cite{mandelstam1991,margolus1998}. When combined, they read
\begin{equation}\label{eq:QSL}
    t_\mathrm{QSL} \ge \frac{\pi}{2} \max \left\{ \frac{1}{\Delta H},\frac{1}{\langle H \rangle} \right\},
\end{equation}
where $\langle H \rangle= \bra{\Phi} H \ket{\Phi}$, $(\Delta H)^2 =  \bra{\Phi}(H-\langle H \rangle )^2 \ket{\Phi}$, and we assume that all eigenvalues of $H$ are positive. In many-body systems, however, we typically have that $\Delta H \sim {n}^{1/2}$ and $\langle H \rangle \sim n$, so that the bound vanishes with system size. A simple consequence of Theorem \ref{th:taylor} gives a significant improvement on the bound.
\begin{corollary} Let $H$ be local and let $\ket{\Phi}$ be a product state. Then, $
    t_\mathrm{QSL} \ge t^*_L$.
\end{corollary}
\begin{proof}
Theorem \ref{th:taylor} shows that the logarithm of the fidelity is analytic for $t < t^*_L$. Since $\log(x)$ is nonanalytic at $x=0$, this means that $\vert \bra{\Phi} e^{-itH} \ket{\Phi} \vert >0$ for $t < t^*_L$.
\end{proof}
Alternatively, the QSL also follows from the explicit lower bound in \eqref{eq:multiplicative}. By truncating the cluster expansion at order $M = 2$, we obtain the lower bound
\begin{equation}
    \vert \bra{\Phi} e^{-itH} \ket{\Phi} \vert \geq \exp\left[ -|S| \frac{(|t|/t_L^*)^4 }{ 1 - (|t|/t_L^*)^2 } \right] e^{- \Delta H^2 t^2 / 2}
    \label{eq:qsl_bound}
\end{equation}
for all $|t| < t_L^*$. The dependence on $|t|/t_L^*$ in the first exponential is better than in Theorem~\ref{th:taylor} because all odd orders in the cluster expansion are purely imaginary and therefore do not contribute to the absolute value $| \exp[\hat{f}_L(t)] |$.

This result shows that the well-known QSLs of \eqref{eq:QSL} do not give very tight bounds for product states evolving under local Hamiltonians. Let us remark that even if the fidelity does not become zero at early times, it does generically quickly become exponentially small in system size \cite{campos_venuti2010,alhambra2020}, which can also be seen from the upper bound in \eqref{eq:multiplicative}.


\section{Summary and Outlook}\label{sec:conclusion}

We showed that the cluster expansion of many dynamical quantities converges at short times, yielding efficient classical approximation algorithms as a by-product. We described the implications for the complexity of quantum dynamics and discussed consequences for concentration bounds, which are linked to the performance of variational quantum algorithms~\cite{DePalmaLimitations,AnshuConc2022}, dynamical phase transitions, and quantum speed limits.

The proof strategy of our main results is based on counting the number of clusters that participate and bounding their individual contributions to the sum. This last step diverges from established convergence proofs of cluster expansions in the literature on abstract polymer models~\cite{KoteckyPreiss,dobrushin1996estimates,Fernandez_2007, friedli2017}, which are based on iterative arguments. We follow more closely recent papers with results on Gibbs states~\cite{kuwahara2020_clustering,kuwahara2020_gaussian,haah2021}, although we use an alternative expression for the cluster derivative involving the Tutte polynomial of the partition graph, which may be of independent interest.

Our work opens the door to many future research directions. It will be interesting to explore the optimality of our algorithms. For example, is it possible to improve the time dependence of Theorem~\ref{th:algoObs} to the one given by the Lieb-Robinson bound (i.e. $e^{\mathcal{O}(t^D)}$ in $D$ dimensions) while retaining the polynomial dependence the inverse approximation error? Similarly, we may ask whether it is possible to extend the concentration bound, Corollary~\ref{co:conc}, to longer times. Related bounds on moments of the distribution have already been shown to hold for times up to $\mathcal{O}(\log n)$~\cite{Moosavian_2022}. One could also address in this context whether a sharp breakdown of Gaussian concentration occurs at longer times, which may be related to dynamical phase transitions.

From a numerical perspective, our algorithms should be compared in practice to existing approaches such as the closely related numerical linked-cluster expansion~\cite{NLCE-Rigol}, which has also been used to approximate quantum dynamics~\cite{QuenchesRigol2014,WhiteHazzard2017,Gan_2020,RichterNLC}.  Other related methods are cluster expansions with tensor-network representations \cite{molnar2015,vanhecke2021} and schemes based on operator-basis expansions~\cite{White2018,Klein_Kvorning_2022,von_Keyserlingk_2022}. While our approach can be adapted to evolution under local Lindbladians by vectorizing the density operator, there is no obvious way in which noise improves convergence of the cluster expansion. It is unclear whether this happens for particular noise models, which could have significant implications on the classical simulation of noisy quantum circuits~\cite{AharonovNoisy,StilckFranca2021,Guillermo} and on the assessment of quantum advantage of noisy, intermediate-scale quantum (NISQ) simulators~\cite{Daley2022}. 

We have shown that the cluster expansion is useful not only for the study of systems in thermal equilibrium, but also for dynamical problems. The cluster expansion enables us to establish the classical approximability of continuous dynamics for short times, similar to previous results for quantum circuits~\cite{Bravyi_2021}. It complements other locality-based methods such as those derived from Lieb--Robinson bounds, which have led to important results in both dynamical and equilibrium systems~\cite{Hastings2010Lieb}. We hope that our work stimulates further research into these techniques and into how they can help us understand other aspects of quantum many-body problems.

\acknowledgements
We thank P.~F.~Wild and J.~I.~Cirac for insightful discussions. AMA acknowledges support from the Alexander von Humboldt foundation. DSW has received funding from the European Union’s Horizon 2020 research and innovation programme under the Marie Skłodowska-Curie Grant Agreement No.~101023276.

\bibliography{bibliography}

\appendix
\widetext

\section{Computation of nested commutators}\label{app:nested}
We show in this appendix that the nested commutator in \eqref{eq:commutator} can be numerically evaluated in time $\exp(\mathcal{O}(M))$. We start by expanding the commutator as
\begin{align}
    &\sum_{\sigma \in S_m} \comm{h_{W_{\sigma(1)}}} {\comm{h_{W_{\sigma(2)}}}{\cdots \comm{h_{W_{\sigma(m)}}}{A}}} \nonumber\\
    & \hspace{4cm} = \sum_{I \subseteq [m]} (-1)^{m-l} \binom{m}{l} \left( \sum_{\sigma \in S_{l}} h_{W_{I_{\sigma(1)}}} \cdots h_{W_{I_{\sigma(l)}}} \right) A \left( \sum_{\sigma \in S_{m-l}}  h_{W_{J_{\sigma(1)}}} \cdots h_{W_{J_{\sigma(m-l)}}} \right),
\end{align}
where the sum over $I$ runs over all subsets of $[m] = \{1, 2, \ldots, m\}$, $l = |I|$, and $J = [m] \setminus I$. The elements of $I$ and $J$ are labeled by $I_1, I_2, \ldots, I_l$ and $J_1, J_2, \ldots, J_{m-l}$ with some arbitrary ordering. Following reference~\cite{mann2021}, we use an inclusion--exclusion argument to rewrite each sum over permutations using the identity
\begin{equation}
    \label{eq:permutation}
    \sum_{\sigma \in S_{l}} h_{W_{{\sigma(1)}}} \cdots h_{W_{{\sigma(l)}}} = \sum_{I \subseteq [l]} (-1)^{l - |I|} \left( \sum_{i \in I} h_{W_i} \right)^l.
\end{equation}
Hence,
\begin{align}
    \sum_{\sigma \in S_m} \comm{h_{W_{\sigma(1)}}} {\comm{h_{W_{\sigma(2)}}}{\cdots \comm{h_{W_{\sigma(m)}}}{A}}} = \sum_{I \subseteq [m]} (-1)^{l} \binom{m}{l} \sum_{I' \subseteq I} \sum_{J' \subseteq J} (-1)^{|I'| + |J'|} \left( \sum_{i \in I'} h_{W_i}\right)^{l} A \left( \sum_{i \in J'} h_{W_i} \right)^{m-l}.
\end{align}
We highlight that the sum over the subsets of $[m]$ involves $2^m$ terms as opposed to the $m!$ terms of the sum over the permutations in $S_m$. For a $k$-local Hamiltonian, the sums over $h_{W_i}$ result in an operator that has support on at most $k m$ spins. The writing down of these operators on the relevant subspace can be achieved in time $\mathrm{poly}(d^{k m}) = \exp(\mathcal{O}(m))$. We can also raise them to the $k^\mathrm{th}$ power with similar computational effort by diagonalizing the operators and powering the eigenvalues. Multiplying the resulting operators by $A$ and $\rho$ and taking the trace again incurs a computational cost with the same asymptotic dependence on $m$. Finally we need to perform the sums over $I$, $I'$, and $J'$. There is a total of $4^m$ terms in these sums such that the overall computational effort indeed scales as $\exp(\mathcal{O}(m))$.


\section{Loschmidt echo}

\subsection{Illustrative example\label{sec:example}}
    In this appendix, we describe the lowest-order terms of the cluster expansion for the Loschmidt echo of a 2-local Hamiltonian. We emphasize that this example serves merely an illustrative purpose. Our results hold for the much broader class of Hamiltonians defined in section~\ref{sec:setup}.

    The spins in this example are arranged on a square lattice as indicated by the black circles in \figref{fig:appendix}(a). The Hamiltonian is assumed to be a sum of terms that act on nearest-neighbor pairs. We note that the Hamiltonian may also include single-spin terms as these can be absorbed into the 2-local interaction terms. Each interaction term of the Hamiltonian is represented in \figref{fig:appendix}(a) by a light-blue diamond placed on the edges of the square lattice. To construct the interaction graph $G$, we connect two diamonds if their associated edges share a spin (dashed lines in \figref{fig:appendix}(a)).

    \begin{figure}[h]
        \centering
        \includegraphics{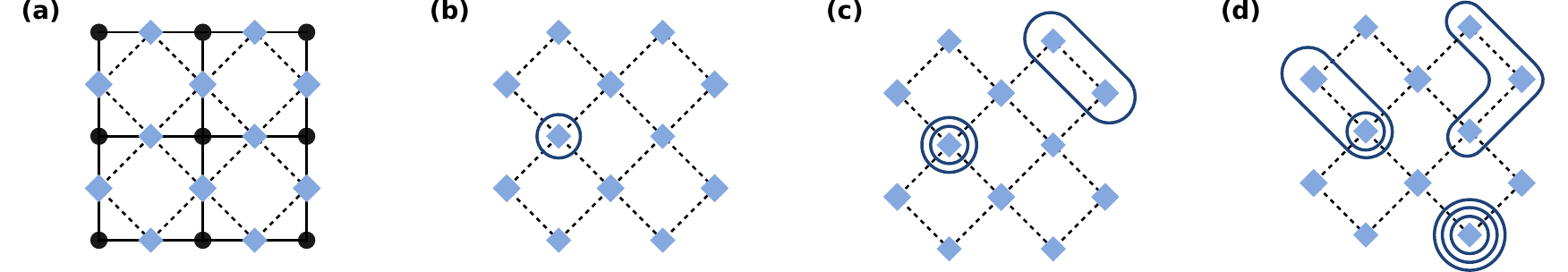}
        \caption{(a)~Spins (black circles) on a square lattice. The nearest-neighbor interaction is indicated by the light-blue diamonds. The dashed lines correspond to edges of the interaction graph. (b)--(d)~Illustrations of all connected clusters of size 1, 2, and 3, respectively, up to rotations and translations. The number of times an interaction term appears in a cluster is determined by how many dark-blue lines enclose it.}
        \label{fig:appendix}
    \end{figure}

    Connected clusters correspond to connected subgraphs of the interaction graphs. All possible connected clusters up to translations and rotations of sizes $1$, $2$, and $3$ are shown in \figref{fig:appendix}(b)--(c).

    To gain intuition for what kind of terms appear in the cluster expansion, we consider the Taylor-series expansion of $\log \langle e^{- i H t} \rangle$ around $t=0$. To third order,
    \begin{equation}
        \log \langle e^{- i H t} \rangle = - i \langle H \rangle t - \frac{1}{2} \left( \langle H^2 \rangle - \langle H \rangle^2 \right) t^2 + \frac{i}{6} \left( \langle H^3 \rangle - 3 \langle H \rangle \langle H^2 \rangle + \langle H \rangle^3 \right) t^3 + O(t^4).
    \end{equation}
    The terms in the series take the form of cumulants. By assumption, the expectation value is with respect to a product state. This leads to a great number of cancellations when substituting in $H = \sum_X \lambda_X h_X$ because expectation values of nonoverlapping terms factorize. For example, the term at second order simplifies to
    \begin{equation}
        \langle H^2 \rangle - \langle H \rangle^2 = \sum_{X} \lambda_X^2 \langle h_X^2 \rangle + \sum_{\substack{X \neq Y\\X \cap Y \neq \emptyset}} \lambda_X \lambda_Y \langle h_X h_y \rangle.
    \end{equation}
    We recognize the two sums as the two distinct types of connected clusters in \figref{fig:appendix}(c), whereas all disconnected clusters cancel. Our formalism of the cluster expansion enables efficient book-keeping of these cancellations at arbitrary order.

\subsection{Proof of Lemma \ref{le:cldev}}\label{app:cldev}
Formally, the cluster expansion of the Loschmidt echo is given by
\begin{equation} 
    L(t) = 1 + \sum_{m \geq 1}  \sum_{\vec{W} \in \mathcal{C}_m} \frac{\lambda^\vec{W}}{\vec{W}!} \mathcal{D}_\vec{W} L(t).
\end{equation}
Note that the sum at this point includes both connected and disconnected clusters.
The cluster derivative of $ L(t)$ can be written as
\begin{equation}
    \label{eq:dev_echo}
    \mathcal{D}_\vec{W} L(t) = (-i t)^{|\vec{W}|} \langle h^\vec{W} \rangle_s,
\end{equation}
where we remind the reader that the symmetrized expectation value is defined by
\begin{equation}
    \langle h^\vec{W} \rangle_s = \frac{1}{m!} \sum_{\sigma \in S_m} \tr \left( h_{{W}_\sigma(1)} h_{{W}_\sigma(2)} \cdots  h_{{W}_\sigma(m)} \rho \right),
\end{equation}
in which $m = |\vec{W}|$. The symmetrized expectation value has the property that it factorizes when $\vec{W}$ is disconnected. Denote by $P_{c, \mathrm{max}}(\vec{W}) \in \mathcal{P}_c(\vec{W})$ the partition of $\vec{W}$ into its maximal connected components. Then,
\begin{equation}
    \langle h^\vec{W} \rangle_s = \prod_{\vec{V} \in P_{c, \mathrm{max}}(\vec{W})} \langle h^\vec{V} \rangle_s
\end{equation}
and
\begin{equation}
    \frac{\lambda^\vec{W}}{\vec{W}!} \mathcal{D}_\vec{W} L(t) = \prod_{\vec{V} \in P_{c, \mathrm{max}}(\vec{W})} \frac{\lambda^\vec{V}}{\vec{V}!} \mathcal{D}_\vec{V} L(t).
\end{equation}
The cluster expansion of the Loschmidt echo becomes
\begin{equation}
    \label{eq:LEpoly}
    L(t) = 1 + \sum_{m \geq 1}  \sum_{\vec{W} \in \mathcal{C}_m} \prod_{\vec{V} \in P_{c, \mathrm{max}}(\vec{W})} \frac{\lambda^\vec{V}}{\vec{V}!} \mathcal{D}_\vec{V} L(t).
\end{equation}

At this point, we take the logarithm, which we again express in terms of its formal Taylor series,
\begin{equation}
    \log (1 + z) = \sum_{n=1}^\infty \frac{(-1)^{n-1}}{n} z^n.
    \label{eq:log}
\end{equation}
Combining \eqref{eq:LEpoly} and \eqref{eq:log} yields
\begin{equation}
    \label{eq:logl_taylor}
    \log L(t) = \sum_{n = 1}^\infty \frac{(-1)^{n-1}}{n} \sum_{m_1, \ldots, m_n \geq 1} \sum_{\substack{\vec{W}_1 \in \mathcal{C}_{m_1}\\ \cdots\\ \vec{W}_n \in \mathcal{C}_{m_n}}} \left( \prod_{\vec{V}_1 \in P_{c, \mathrm{max}}(\vec{W}_1)} \frac{\lambda^{\vec{V}_1}}{\vec{V}_1!} \mathcal{D}_{\vec{V}_1} L(t) \right) \cdots \left( \prod_{\vec{V}_n \in P_{c, \mathrm{max}}(\vec{W}_n)}  \frac{\lambda^{\vec{V}_n}}{\vec{V}_n!} \mathcal{D}_{\vec{V}_n} L(t) \right).
\end{equation}
We can rearrange the sums to first sum over all clusters $\vec{W}$ before considering decompositions of $\vec{W}$ into connected clusters $\vec{V}$:
\begin{equation}
    \log L(t) =  \sum_{m \geq 1} \sum_{\vec{W} \in \mathcal{G}_m} \sum_{P \in \mathcal{P}_c(\vec{W})} C(P) \prod_{\vec{V}\in P}\frac{\lambda^\vec{V}}{ \clus{V}  !} \mathcal{D}_\vec{V} L(t) .
\end{equation}
Lemma~\ref{le:connected} allowed us to impose that $\vec{W}$ be connected. The coefficient $C(P)$ can be determined by considering the different ways in which the partition $P$ can arise from the clusters $\vec{W}_1, \vec{W}_2, \ldots, \vec{W}_n$ in \eqref{eq:logl_taylor}, such that $P = \bigcup_{i = 1}^n P_{c, \mathrm{max}}(\vec{W}_i)$. Different elements of $P$ can belong to the same ``parent'' cluster $\vec{W}_i$ if they do not overlap. It is hence possible to construct assignments of all $\vec{V} \in P$ to parent clusters $\vec{W}_1, \vec{W}_2, \ldots, \vec{W}_n$ from a proper coloring of the partition graph $\tilde{G}_P$ with exactly $n$ colors. The vertices colored with the first color form $\vec{W}_1$, the second color gives $\vec{W}_2$, and so on. From this argument, we find that
\begin{equation}
    C(P) = \frac{1}{P!} \sum_{n = 1}^{|P|} \frac{(-1)^{n-1}}{n} \chi^*_{\tilde{G}_P}(n),
\end{equation}
where $\chi^*_{\tilde{G}_P}(n)$ is the number of proper colorings of $\tilde{G}_P$ with exactly $n$ colors. The combinatorial factor $P!$ removes overcounting that occurs when $P$ contains clusters with multiplicity greater than $1$ because permuting the colors of repeated clusters has no effect on $\vec{W}_1, \vec{W}_2, \ldots \vec{W}_n$.

To complete the proof of Lemma~\ref{le:cldev}, we make use of the following combinatorial property of graphs, which we prove in Appendix~\ref{sec:tutte}.
\begin{lemma}
    \label{le:tutte}
    Given a connected graph $G = (V, E)$, we denote by $\chi^*_G(n)$ the number of proper colorings of $G$ with exactly $n$ colors. Let $T_G(x, y)$ be the Tutte polynomial of $G$. Then,
    \begin{equation}
        \sum_{n = 1}^{|V|} \frac{(-1)^{n-1}}{n} \chi^*_{G}(n) = (-1)^{|V| - 1} T_{G}(1, 0).
    \end{equation}
\end{lemma}
By applying this lemma to $C(P)$, we obtain
\begin{equation}
    \label{eq:logl}
    \log L(t) =  \sum_{m \geq 1} \sum_{\vec{W} \in \mathcal{G}_m} \sum_{P \in \mathcal{P}_c(\vec{W})} \frac{(-1)^{|P|-1}}{P!} T_{\tilde{G}_P}(1,0) \prod_{\vec{V}\in P}\frac{\lambda^\vec{V}}{ \clus{V}  !} \mathcal{D}_\vec{V} L(t) .
\end{equation}
Finally, taking the cluster derivative of \eqref{eq:logl} yields
\begin{equation}\label{eq:clusdev2}
    \mathcal{D}_\vec{W} \log L(t) = \vec{W}! \sum_{P \in \mathcal{P}_c(\vec{W})} \frac{(-1)^{\vert P \vert -1}}{P!}T_{\tilde G_P}(1,0) \prod_{\vec{V}\in P}\frac{1}{ \clus{V}  !} \mathcal{D}_\vec{V} L(t)
\end{equation}
which, using \eqref{eq:dev_echo}, can be readily brought into the form of the expression in Lemma~\ref{le:cldev}.

\subsection{Proof of Lemma~\ref{le:tutte}\label{sec:tutte}}
To prove Lemma~\ref{le:tutte}, we introduce three more short lemmas. The first one relates the number of colorings that use exactly $k$ colors to the chromatic polynomial.
\begin{lemma}
    \label{le:chromatic}
    Given a graph $G$, let $\chi_G^*(k)$ denote the number of proper colorings of $G$ that use exactly $k$ colors. Moreover, let $\chi_G(k)$ be the chromatic polynomial, that is, the number of proper colorings with up to $k$ colors. Then,
    \begin{equation}
        \label{eq:chromatic}
        \chi_G^*(n) = \sum_{k = 1}^n (-1)^{n - k} \binom{n}{k} \chi_G(k).
    \end{equation}
\end{lemma}
\begin{proof}
    This lemma follows from a standard inclusion--exclusion argument. Alternatively, we can prove the statement by direct calculation as follows.
    
    The chromatic polynomial $\chi_G(k)$ can be computed by picking $j \leq k$ colors and adding the contributions from $\chi_G^*(j)$:
    \begin{equation}
        \chi_G(k) = \sum_{j = 1}^k \binom{k}{j} \chi_G^*(j)
    \end{equation}
    We substitute this expression into the right-hand side of \eqref{eq:chromatic} and exchange the order of the sums:
    \begin{equation}
        \sum_{k = 1}^n (-1)^{n - k} \binom{n}{k} \chi_G(k) = \sum_{k = 1}^n \sum_{j = 1}^k (-1)^{n - k} \binom{n}{k} \binom{k}{j} \chi_G^*(j) = \sum_{j = 1}^n \chi_G^*(j) \sum_{k = j}^n (-1)^{n - k} \binom{n}{k} \binom{k}{j}.
    \end{equation}
    When $j = n$, the last sum evaluates to $1$. For $j < n$, we find instead
    \begin{equation}
        \sum_{k = j}^n (-1)^{n - k} \binom{n}{k} \binom{k}{j} =  \binom{n}{j}  \sum_{k = j}^n (-1)^{n-k} \binom{n-j}{k-j} = \binom{n}{j}  \sum_{k = 0}^{n-j} (-1)^{n-j-k} \binom{n-j}{k} = 0.
    \end{equation}
    Thus,
    \begin{equation}
        \sum_{k = 1}^n (-1)^{n - k} \binom{n}{k} \chi_G(k) = \chi_G^*(n)
    \end{equation}
    as claimed in the lemma.
\end{proof}

The second lemma connects the chromatic polynomial to the Tutte polynomial. A proof of this statement can be found in, e.g., Ref.~\cite{Brylawski1992}.
\begin{lemma}
    \label{le:chrom_tutte}
    Given a graph $G = (V, E)$ with $c$ connected components, the chromatic polynomial $\chi_G(k)$ is related to the Tutte polynomial $T_G(x, y)$ by
    \begin{equation}
        \chi_G(k) = (-1)^{|V| - c} k^c T_G(1 - k, 0).
    \end{equation}
\end{lemma}

The third lemma is a simple identity involving sums over binomial coefficients and powers of integers.
\begin{lemma}
    \label{le:binom}
    The following identity holds for any integers $n$ and $k$ satisfying $n > k \geq 0$.
    \begin{equation}
        \sum_{j = 1}^n (-1)^{n-j} \binom{n}{j} j^k = (-1)^{n-1} \delta_{k, 0}.
    \end{equation}
\end{lemma}
\begin{proof}
    For $k = 0$,
    \begin{equation}
        \sum_{j=1} (-1)^{n-j} \binom{n}{j} = (-1)^{n-1} + \sum_{j=0}^n (-1)^{n-j} \binom{n}{j} = (-1)^{n-1}.
    \end{equation}
    For $n > k > 0$, we observe that
    \begin{equation}
        \sum_{j = 1}^n (-1)^{n-j} \binom{n}{j} j^k x^j = \left( x \frac{\di}{\di x} \right)^k (x - 1)^n.
    \end{equation}
    The lemma follows from the fact that the right-hand side vanishes at $x=1$ for all $n > k > 0$.
\end{proof}

\begin{proof}[Proof of Lemma~\ref{le:tutte}]
    We apply the above three lemmas in order. From Lemma~\ref{le:chromatic}, we have
    \begin{equation}
        \sum_{n=1}^{|V|} \frac{(-1)^{n-1}}{n} \chi_G^*(n) = \sum_{n=1}^{|V|} \frac{(-1)^{n-1}}{n} \sum_{k=1}^n (-1)^{n-k} \binom{n}{k} \chi_G(k).
    \end{equation}
    We switch the order of the sums to obtain
    \begin{equation}
        \sum_{n=1}^{|V|} \frac{(-1)^{n-1}}{n} \chi_G^*(n) = \sum_{k=1}^{|V|} \sum_{n=k}^{|V|} \frac{(-1)^{k-1}}{n} \binom{n}{k} \chi_G(k) =  \sum_{k=1}^{|V|} \frac{(-1)^{k-1}}{k} \chi_G(k) \sum_{n=k}^{|V|} \binom{n-1}{k-1}.
    \end{equation}
    By the hockey-stick identity, the last sum evaluates to $\binom{|V|}{k}$. Combined with Lemma~\ref{le:chrom_tutte}, setting $c = 1$ since $G$ is connected by assumption, this yields
    \begin{equation}
        \label{eq:almost_there}
        \sum_{n=1}^{|V|} \frac{(-1)^{n-1}}{n} \chi_G^*(n) = \sum_{k=1}^{|V|} \binom{|V|}{k} (-1)^{|V|-k} T_G(1-k, 0).
    \end{equation}
    The Tutte polynomial $T_G(x, 0)$ is a polynomial in $x$ of degree at most $|V|-1$. Therefore, $T_G(1-k, 0)$ is a polynomial in $k$ of the same maximum degree, which allows us to write
    \begin{equation}
        T_G(1-k, 0) = \sum_{n = 0}^{|V|-1} a_n k^n
    \end{equation}
    for some coefficients $a_n$. By substituting into \eqref{eq:almost_there}, we get
    \begin{equation}
        \sum_{n=1}^{|V|} \frac{(-1)^{n-1}}{n} \chi_G^*(n) = \sum_{n=0}^{|V|-1} a_n \sum_{k=1}^{|V|} \binom{|V|}{k} (-1)^{|V|-k} k^n.
    \end{equation}
    Lemma~\ref{le:binom} allows us to simplify the expression to
    \begin{equation}
        \sum_{n=1}^{|V|} \frac{(-1)^{n-1}}{n} \chi_G^*(n) = (-1)^{|V|-1} a_0.
    \end{equation}
    This is the desired expression since $a_0 = T_G(1, 0)$.
\end{proof}


\subsection{Proof of Proposition~\ref{le:cldev_bound}}\label{app:cldev_bound}

We start from the expression for the cluster derivative in Lemma~\ref{le:cldev}:
\begin{align}
        \mathcal{D}_\vec{W} & \log  L(t)  = (-i t)^m
         \sum_{P \in \mathcal{P}_c(\vec{W})} (-1)^{|P|-1} N_P(\vec{W}) T_{\tilde{G}_P}(1, 0) \prod_{\vec{V} \in P} \langle h^\vec{V} \rangle_s. \nonumber
\end{align}
Recall the symmetric expectation value
$
        \langle h^\vec{V} \rangle_s =  \frac{1}{|\vec{V}|!} \sum_{\sigma \in S_{|\vec{V}|}} \tr \left( h_{V_{\sigma(1)}} h_{V_{\sigma(2)}} \cdots h_{V_{\sigma(|\vec{V}|)}} \rho \right)
$
and the combinatorial factor $N_P(\vec{W}) = {\vec{W}!}/\left(P! \prod_{\vec{V} \in P} \vec{V}! \right)$. We rewrite the sum over cluster partitions of $\vec{W}$ as a sum of graph partitions of the cluster graph $G_{\vec{W}}$. Here, a graph partition refers to a partition of the vertices. We only consider partitions into connected subgraphs, meaning that the subsets of vertices in the partition induce connected subgraphs on the original graph. For every graph partition of $G_\vec{W}$ into connected subgraphs, there exists exactly one corresponding cluster partition $P \in \mathcal{P}_c(\vec{W})$.  On the other hand, for every cluster partition there are exactly $N_P(\clus{W})=\vec{W}!/ P! \prod_{\vec{V}\in P} \clus{V}! $ equivalent graph partitions. This combinatorial factor arises because repeated subsystems are indistinguishable at the level of the cluster but give rise to distinct vertices in the cluster graph. With this,
\begin{align}
    \mathcal{D}_\vec{W} \log  L(t) = (-i t)^m \sum_{P' \in \mathcal{P}_c(G_\vec{W})} (-1)^{|P'|-1} T_{\tilde{G}_{P}}(1, 0) \prod_{\vec{V} \in P} \langle h^\vec{V} \rangle_s,
\end{align}
where, in a slight abuse of notation, $\mathcal{P}_c(G_\vec{W})$ is the set of partitions of the cluster graph $G_\vec{W}$ into connected components and $P \in \mathcal{P}_c(\vec{W})$ is the cluster partition corresponding to $P'$. By observing that $\vert \langle h^\vec{V} \rangle_s \vert \le 1$, we obtain
\begin{equation} \label{eq:clderbound1}
    \vert  \mathcal{D}_\vec{W} \log  L(t) \vert \le \vert t \vert^m  \sum_{P' \in \mathcal{P}_c(G_\vec{W})}  T_{\tilde{G}_{P}}(1, 0) \le  \vert t \vert^m  \sum_{{P'} \in \mathcal{P}_c(G_\vec{W})}  T_{\tilde{G}_{P}}(1, 1),
\end{equation}
where the last inequality follows from the fact that the Tutte polynomial $T_{\tilde{G}_P}(x, y)$ has positive coefficients. We bound this sum in two steps, starting with the following lemma.

\begin{figure}[b]
    \centering
    \includegraphics{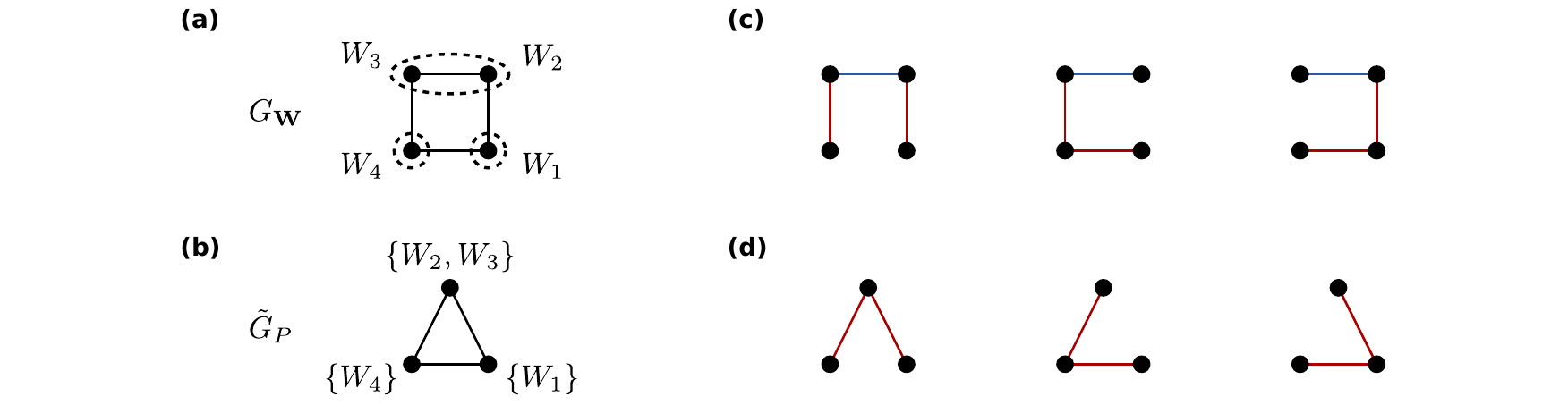}
    \caption{(a)~A cluster graph $G_\vec{W}$ with a particular graph partition indicated by the dashed outlines. (b)~The corresponding partition graph. (c)~The three bicolored spanning trees that we identify with this graph partition. (d)~The red edges define a spanning tree on the partition graph.}
    \label{fig:lemma21}
\end{figure}

\begin{lemma} \label{le:parttotrees}
    Given a connected cluster $\vec{W} \in \mathcal{G}_m$,
    \begin{equation}
        \sum_{P' \in \mathcal{P}_c(G_\clus{W})} T_{\tilde{G}_{P}}(1, 1) \leq 2^m T_{G_\clus{W}}(1, 1).
    \end{equation}
\end{lemma}
\begin{proof}
$T_{\tilde{G}_{P'}}(1, 1)$ counts the number of spanning trees of $\tilde{G}_{P}$, so that 
\begin{equation}\label{eq:spantree}
       \sum_{P' \in \mathcal{P}_c(G_\clus{W})} T_{\tilde{G}_{P}}(1, 1) = \sum_{P' \in \mathcal{P}'_c(G_\clus{W})} \sum_{\substack{\text{spanning} \\ \text{trees of } \tilde{G}_{P}}} 1.
\end{equation}
Given a spanning tree of the cluster graph $G_{\clus{W}}$, consider a bicoloring of the edges into blue and red, and delete the red ones. This separates the edges into disconnected components, each of which induces a connected subgraph on $G_{\clus{W}}$. These subgraphs define a partition $P' \in \mathcal{P}_c(G_\vec{W}) $, with its corresponding partition graph $\tilde{G}_{P}$. Moreover, the deleted red edges can be identified with a spanning tree of $\tilde{G}_{P}$. Hence, any bicoloring of a spanning tree of $G_\vec{W}$ describes a term in the double sum on the right-hand side of \eqref{eq:spantree}. Conversely, for every term in the sum, we can find at least one bicolored spanning tree. This procedure is illustrated in \figref{fig:lemma21}. It follows that the sum is bounded from above by the number of bicolored trees of $G_{\clus{W}}$. The number of edges of each spanning tree is $\vert \clus{W}\vert -1=m-1$, so that there are always $2^{m-1}$ distinct bicolorings for every tree.  The total number of spanning trees is $T_{G_\clus{W}}(1, 1)$, which completes the proof.
\end{proof}

We next bound the number of spanning trees of $T_{G_\clus{W}}(1, 1)$ given the connectivity of the cluster graph.
\begin{lemma}
    \label{le:trees}
    Consider a set of supports $S$ such that the maximum degree of the associated interaction graph $G$ is $\mathfrak{d}$. Given a connected cluster $\vec{W} \in \mathcal{G}_m$,
    \begin{equation}
        \frac{1}{\clus{W}!} T_{G_\clus{W}}(1, 1) \leq [e( \mathfrak{d} + 1)]^{m+1} .
    \end{equation}
\end{lemma}
\begin{proof}
    Let us fix $W_1 \in \vec{W}$ as the root of a spanning tree. Any spanning tree can be constructed by picking for each vertex $W_2, W_3, \ldots W_m$ one of the edges incident on it. The number of choices is $\gamma(W_2) \gamma(W_3) \cdots  \gamma(W_m)$, where $\gamma(W_k)$ is the degree of the vertex associated with $W_k$ in $G_\clus{W}$. It follows that
     \begin{equation}
        T_{G_\clus{W}}(1, 1) \leq \gamma(W_1) \gamma(W_2) \cdots  \gamma(W_m).
    \end{equation}
    
    This product of degrees can be bounded as in the proof of Proposition 3.8 in \cite{haah2021}. The degree of a vertex associated with $W_k$ can be written as
    \begin{equation}
       \gamma(W_k)= \mu_{\clus{W}}{(W_k)} -1 + \sum_{X \in \mathcal{N}(W_k)} \mu_{\clus{W}}{(X)},
    \end{equation}
    where $\mathcal{N}(W_k)$ is the set of neighbors of $W_k$ in the interaction graph $G$. We sum the degrees over all distinct subsystems that appear in $\clus{W}$:
    \begin{equation}
        \sum_{X \in S \, : \, \mu_\vec{W}(X) > 0} \gamma(X) = \sum_{X \in S \, : \, \mu_\vec{W}(X) > 0} \left( \mu_{\clus{W}}{(X)} - 1 + \sum_{Y \in \mathcal{N}(X)} \mu_{\clus{W}}{(Y)} \right)  \le m-1 + \mathfrak{d}m.
    \end{equation}
    The double sum is bounded by $\mathfrak{d} m$ because each $\mu_\vec{W}(Y)$ appears in it at most $\mathfrak{d}$ times. Finally, we bound
    \begin{align}
        \frac{1}{\clus{W}!}  \gamma(W_1) \gamma(W_2) \cdots  \gamma(W_m) &= \prod_{X \in S \, : \, \mu_\vec{W}(X) > 0} \frac{\left( \mu_{\clus{W}}(X) - 1 + \sum_{Y \in \mathcal{N}(X)} \mu_{\clus{W}}{(Y)} \right)^{\mu_{\clus{W}}(X)} }{\mu_{\clus{W}}(X)!} \\
        & \le e^m  \prod_{X \in S \, : \, \mu_\vec{W}(X) > 0} \left( \frac{ \mu_{\clus{W}}(X) -1 + \sum_{Y \in \mathcal{N}(X)} \mu_{\clus{W}}{(Y)}  }{\mu_{\clus{W}}(X)} \right)^{\mu_{\clus{W}}(X)} \\ 
        & \le e^m \left( \frac{m - 1 + \mathfrak{d} m}{m} \right)^m \le [ e(\mathfrak{d} + 1) ]^m.
    \end{align}
\end{proof}
The bound on the cluster derivative follows from the two lemmas and \eqref{eq:clderbound1}.


\section{Generalized Loschmidt echo}\label{app:generalizedL2}
In this appendix, we prove Theorem \ref{th:taylorL}. To this end, we first establish \eqref{eq:cld}. We again start from the formal cluster expansion of the Loschmidt echo:
\begin{equation}   
    L(\{t_l\}) = 1 + \sum_{m \geq 1}  \sum_{\vec{W} \in \mathcal{C}^K_m} \frac{\lambda^\vec{W}}{\vec{W}!} \mathcal{D}_\vec{W} L(\{t_l\}).
\end{equation}
The cluster derivative takes the more complicated form
\begin{equation}
    \label{eq:dwl}
    \mathcal{D}_\vec{W} L(\{t_l\}) = \left[ \prod_{l=1}^K \frac{(-i t_l)^{m_l}}{m_l!} \right] \tr \left[ \left( \sum_{\sigma \in S_{m_1}} h_{W_{1, \sigma(1)}}^{(1)} \cdots h_{W_{1, \sigma(m_1)}}^{(1)} \right) \cdots \left( \sum_{\sigma \in S_{m_K}} h_{W_{K, \sigma(1)}}^{(K)} \cdots h_{W_{K, \sigma(m_K)}}^{(K)} \right) \rho \right].
\end{equation}
Here, $m_l = |\vec{W}_l|$ and $\vec{W}_l$ are the parts of $\vec{W}$ associated with the Hamiltonian $H^{(l)}$. Despite these complications, one can readily check that we still have
\begin{equation}
    \frac{\lambda^\vec{W}}{\vec{W}!} \mathcal{D}_\vec{W} L(\{t_l\}) = \prod_{\vec{V} \in P_{c, \mathrm{max}}(\vec{W})} \frac{\lambda^\vec{V}}{\vec{V}!} \mathcal{D}_\vec{V} L(\{t_l\}).
\end{equation}
Here, the components are disconnected if their subsystems do not overlap, irrespective of the Hamiltonian label. The remaining arguments from Appendix~\ref{app:cldev} carry over and we obtain
\begin{equation}
    \frac{1}{\vec{W}!} \mathcal{D}_\vec{W} \log L(\{t_l\}) = \sum_{P \in \mathcal{P}_c(\vec{W})} \frac{(-1)^{\vert P \vert -1}}{P!}T_{\tilde G_P}(1,0) \prod_{\vec{V}\in P}\frac{1}{ \clus{V}  !} \mathcal{D}_\vec{V} L(\{t_l\}).
\end{equation}
By using the fact that $|\mathcal{D}_\vec{W} L(\{t_l\})| \leq \prod_{l=1}^K |t_l|^{m_l}$ and following the steps in Appendix~\ref{app:cldev_bound}, we arrive at an upper bound for the cluster derivative. The only change is that the relevant interaction graph is $G^K$, which has maximum degree $K(\mathfrak{d}+1) - 1$ given the maximum degree $\mathfrak{d}$ of $G$. For $\vec{W} \in \mathcal{G}_m^K$, this yields \eqref{eq:cld}.

To prove the convergence statement in Theorem~\ref{th:taylorL}, we observe that
\begin{equation}
    \sum_{\vec{W} \in \mathcal{G}_m^K} \prod_{l = 1}^K |t_l|^{m_l} \leq \sum_{\vec{W} \in \mathcal{G}_m} \sum_{\substack{m_1, \ldots, m_K \geq 0\\ m_1 + \cdots m_K = m}} \frac{m!}{m_1! \cdots m_K!} \prod_{l = 1}^K |t_l|^{m_l} = \left( \sum_{l = 1}^K |t_l| \right)^m | \mathcal{G}_m| \leq  |S| \left( e \mathfrak{d} \sum_{l = 1}^K |t_l| \right)^m,
\end{equation}
where we used the bound on $|\mathcal{G}_m|$ from Lemma~\ref{le:clusters}. Combining these results yields
\begin{align}
   \left| \sum_{\vec{W} \in \mathcal{G}_m^K} \frac{\lambda^\vec{W}}{\vec{W}!} \mathcal{D}_\vec{W} \log L(t) \right| \leq \vert S \vert  \left[2 e^2 K \mathfrak{d} (\mathfrak{d}+1) \sum_{l=1}^K |t_l| \right]^{m}.
\end{align}
The convergence of the cluster expansion for this generalized Loschmidt echo follows in an analogous fashion to Theorem~\ref{th:taylor}.


To prove the bound on the computational cost in Theorem~\ref{th:taylorL}, we closely follow the proofs of Proposition~\ref{le:algo} and Theorem~\ref{th:computation}, keeping in mind that the degree of the relevant interaction graph is $K (\mathfrak{d} + 1) - 1$. Throughout, we only keep the dependence on $K$ explicit, while suppressing the dependence on $\mathfrak{d} = \mathcal{O}(1)$.

The computational cost of step (i) of the proof of Proposition~\ref{le:algo} is modified to $\exp(\mathcal{O}(m \log K))$. Step (ii) remains unchanged as the computational cost of evaluating the Tutte polyonomial only depends on the number of vertices. In step (iii), we have to evaluate \eqref{eq:dwl} instead of the simpler symmetrized expectation value. Nevertheless, by rewriting the sums over permutations using \eqref{eq:permutation}, this can still be carried out in time $\exp(\mathcal{O}(n))$. Hence, the cluster derivative of the generalized Loschmidt echo can be computed in time $\exp(\mathcal{O}(m \log K))$. By Lemma~\ref{le:clusters}, the run time of the algorithm to enumerate the clusters is $|S| \times \exp(\mathcal{O}(m \log K))$. We conclude that the truncated cluster expansion of the generalized Loschmidt echo that includes all clusters up to size $M$ can be computed with a total run time $|S| \times \exp(\mathcal{O}(M \log K))$. Choosing $M > \log\frac{\vert S \vert }{\left(1-K \sum_l \vert t_l \vert / t_L^* \right)\varepsilon} /\log \frac{t_L^*}{K\sum_l \vert t_l \vert}$ completes the proof.


\section{Proof of concentration bound} \label{app:concentration}

Let us assume $\tr(\rho e^{it H^{(1)}} H^{(2)} e^{- it H^{(1)}}   )=0$ for simplicity. We write the cluster expansion  of $ \log L(\{ t,\lc,-t \}) \equiv \log \tr(\rho e^{it H^{(1)}} e^{\lc H^{(2)}} e^{- it H^{(1)}})$ as
\begin{equation}
 f_L(\{ t,\lc,-t \})=   \log \tr(\rho e^{it H^{(1)}} e^{\lc H^{(2)}} e^{- it H^{(1)}}) = \sum_{m \geq 1}  \sum_{\vec{W}  \in \mathcal{G}^3_m} \frac{\lambda^{\vec{W} }}{\vec{W}!} \mathcal{D}_{\vec{W}} \log L(\{ t,\lc,-t \}),
\end{equation}
where we have three Hamiltonians. However, note that since $f_L(\{ t,0,-t \})=0$, only the clusters with at least two terms $h^{(2)}_X$ from $H^{(2)}$ contribute. This means that the smallest power of $\lc$ is $2$. That is,
\begin{equation}
     f_L(\{ t,\lc,-t \}) = \sum_{m \geq 2} \sum_{\vec{W}  \in \mathcal{G}^3_m: \, m_2 \ge 2} \frac{\lambda^{\vec{W} }}{\vec{W}!} \mathcal{D}_{\vec{W}} \log L(\{ t,\lc,-t \}),
\end{equation}
where $m_2$ is the number of subsystems in $\vec{W}$ associated with $H^{(2)}$. Now, given \eqref{eq:cld} and the bound on the number of clusters in Lemma~\ref{le:clusters}, we have that
\begin{align}
     \left| f_L(\{ t,\lc,-t \}) \right| & \le  \sum_{m \geq 2} \sum_{\vec{W}  \in \mathcal{G}^3_m: \, m_2 \ge 2} \vert t \vert^{m-m_2} \lc^{m_2} (6 e  (\mathfrak{d} + 1))^m
   \\  & \le \vert S \vert \sum_{m \ge 2} (6 e   (\mathfrak{d} + 1))^m  (e\mathfrak{d})^{m} \left ( \sum_{\substack{m_1, m_3 \geq 0, \, m_2 \geq 2\\m_1 + m_2 + m_3 = m}}  \frac{m!}{m_1! m_2! m_3!}  \vert t \vert^{m-m_2} \lc^{m_2} \right )
   \\ & \le \vert S \vert \sum_{m \ge 2} (6 e^2   \mathfrak{d}(\mathfrak{d} + 1))^m  \lc^2 m^2    \left ( \sum_{\substack{m_1, m_2, m_3 \geq 0\\m_1 + m_2 + m_3 = m-2}}  \frac{(m-2)!}{m_1! m_2! m_3!}  \vert t \vert^{m-m_2} \lc^{m_2} \right )
    \\ & = \vert S \vert \sum_{m \ge 2} (6 e^2   \mathfrak{d}(\mathfrak{d} + 1))^m  \lc^2 m^2 (2 \vert t \vert + \lc)^{m-2},
\end{align}
where in the second line we write the sum over combinations with at least two terms from $H^{(2)}$ (that is, the number of ways in which one can arrange $m$ objects in three boxes, with at least two objects in one of them). In the third line, we bound this sum by $m^2$ times the number of ways of arranging $m-2$ elements in three boxes. Finally, we use $\sum_{m \ge 2} m^2 a^{m-2} =\frac{4+a^2-3a}{(1-a)^3} $ and the definition of $t^*_L$ in Theorem \ref{th:taylor} to obtain
\begin{align}
    f_L(\{ t,\lc,-t \}) & \le  \vert S \vert \left( \frac{3 \lc}{t^*_L} \right)^2 \frac{4+(3(2 \vert t \vert + \lc) / t_L^*)^2}{\left[1-3(2 \vert t \vert + \lc) / t_L^* \right]^3} .
\end{align}

Now take $\vert t \vert \le t^*_L /7$, so that $1-\frac{3(2\vert t \vert + \lc)}{t^*_L} \ge \frac{1}{7}- \frac{3 \lc}{t_L^*}$. We obtain 
\begin{equation}
   \log  \tr(\rho(t) e^{\lc H^{(2)}}) \le 5 \vert S \vert \frac{(3 \lc /t^*_L)^2}{\left(1/7 -{3 \lc}/{t^*_L}\right)^3}.
\end{equation}
Choosing $\lc =  \eta \delta (t^*_L)^2 / \vert S \vert $ with $\eta = 1/(90 \times7^3)$ then leads to
\begin{equation}
  e^{-\delta \lc}  \tr \left[ \rho(t) e^{\lc H^{(2)}} \right] \le e^{- {(\delta t^*_L)^2}/{(250)^2 \vert S \vert}},
\end{equation}
under the condition that $ \delta \le \frac{ \vert S \vert}{t_L^*}$, which is trivially satisfied since $\delta \le \vert S \vert$ and $t_L^* <1$. Together with the inequality
\begin{equation}
    \text{Pr} [\vert x - \tr[ \rho (t) H^{(2)}] \vert \ge \delta] \le 2  e^{-\delta \lc}  \tr(\rho(t) e^{\lc H^{(2)}}),
\end{equation}
this proves the result.


\section{Proof of Corollary \ref{co:DPTs}}\label{app:DPTs}

From \eqref{eq:expansion} we have the series expansion
\begin{equation}
   \frac{g_{n}(t)}{n} = \frac{1}{n} \left( \sum_{m=1}^\infty \sum_{\vec{W} \in \mathcal{G}_m} \frac{\lambda^\vec{W}}{\vec{W}!} \mathcal{D}_\vec{W} \log L_{n}(t) \right).
\end{equation}
Theorem \ref{th:taylor} shows that $g_{n}(t)$ is analytic for $t < t^*_L$, and any finite $n$, such that
\begin{equation}
    \left| \sum_{\vec{W} \in \mathcal{G}_m} \frac{\lambda^\vec{W}}{n \times \vec{W}!} \mathcal{D}_\vec{W} \log L_{n}(t) \right| \leq \frac{e \mathfrak{d} \vert S \vert}{(1+e(\mathfrak{d}-1)) n} ( |t|/t^*_L)^{m}.
\end{equation}
For local Hamiltonians, $\frac{ \vert S \vert}{n} \equiv C = \mathcal{O}(1)$, and thus
\begin{equation}
   \sum_{m=m_0}^\infty \left(\frac{t}{t^*_L}\right)^m C e \mathfrak{d} < \infty.
\end{equation}
We apply Tannery's theorem to the sequence $\{ \sum_{\vec{W} \in \mathcal{G}_m} \frac{\lambda^\vec{W}}{n \vec{W}!} \mathcal{D}_\vec{W} \log L_{n}(t) \}$. This implies that we can place the limit inside the sum as
\begin{equation}
    G(t)= \sum_{m=0}^{\infty} \lim_{n \rightarrow \infty}  \sum_{\vec{W} \in \mathcal{G}_m} \frac{\lambda^\vec{W}}{n \vec{W}!} \mathcal{D}_\vec{W} \log L_{n}(t).
\end{equation}
That $G(t)$ is analytic follows from the fact that $\lim_{ n \rightarrow \infty}  \sum_{\vec{W} \in \mathcal{G}_m} \frac{\lambda^\vec{W}}{n \vec{W}!} \mathcal{D}_\vec{W} \log L_{n}(t) \le (\frac{t}{t^*_L})^{m} C e \mathfrak{d}$.

\end{document}